\documentclass[11pt]{article}
\usepackage{amsmath}
\usepackage{graphicx,psfrag,epsf}
\usepackage{enumerate}
\usepackage{psfrag,epsf}
\usepackage{url} 
\usepackage{amsmath,amssymb,amsfonts,amsthm,mathtools}
\usepackage{bm,bbm,mathrsfs}
\usepackage{color,xcolor}
\usepackage{subfigure}
\usepackage{algorithm,algpseudocode}

\usepackage{scalefnt,bm}
\usepackage{multirow}
\usepackage{authblk} 
\usepackage[colorlinks=true, citecolor=blue]{hyperref}

\usepackage[arrow]{xy}

\usepackage{natbib}
\usepackage{dsfont}
\usepackage[title]{appendix}
\usepackage{lmodern}
\usepackage[title]{appendix}

\usepackage[left=1in,top=1.1in,right=0.5in,bottom=1in]{geometry}

\usepackage{letltxmacro}
\makeatletter
\let\oldr@@t\r@@t
\def\r@@t#1#2{%
	\setbox0=\hbox{$\oldr@@t#1{#2\,}$}\dimen0=\ht0
	\advance\dimen0-0.2\ht0
	\setbox2=\hbox{\vrule height\ht0 depth -\dimen0}%
	{\box0\lower0.4pt\box2}}
\LetLtxMacro{\oldsqrt}{\sqrt}
\renewcommand*{\sqrt}[2][\ ]{\oldsqrt[#1]{#2}}
\makeatother

\theoremstyle{definition}

\newtheorem{theorem}{Theorem}[section]

\newtheorem{corollary}[theorem]{Corollary}

\newtheorem{definition}{Definition}[section]

\newtheorem{proposition}[theorem]{Proposition}
\newtheorem{remark}[theorem]{Remark}
\numberwithin{equation}{section} 

\makeatletter
\def\@seccntformat#1{\@ifundefined{#1@cntformat}%
	{\csname the#1\endcsname\quad}
	{\csname #1@cntformat\endcsname}
}
\makeatother

\newif\ifShowComments
\ShowCommentstrue
\def\strutdepth{\dp\strutbox}
\def\druk#1{\strut\vadjust{\kern-\strutdepth
        {\vtop to \strutdepth{%
                \baselineskip\strutdepth\vss
                        \llap{\hbox{#1}\quad}\null}}}}

\title{\bf
\huge
Family of multivariate extended skew-elliptical distributions: Statistical
properties, inference and application
}
\author{
\text{Roberto Vila}$^{1}$\thanks{Corresponding author: Roberto Vila, email: {rovig161@gmail.com}
\newline
}\,\,,
\text{Helton Saulo}$^{1,2}$,
\text{Leonardo Santos}$^{1}$,
\text{João Monteiros}$^{1}$
and
\text{Felipe Quintino}$^{1}$\\
{\small $^{1}$ Department of Statistics, University of Brasilia, Brasilia, Brazil}\\
{\small $^{2}$ Department of Economics, Federal University of Pelotas, Pelotas, Brazil}\\
}

\begin{document}
\maketitle
\begin{abstract}
In this paper we propose a family of multivariate asymmetric distributions over an arbitrary subset of set of real numbers  which is defined in terms of the well-known elliptically symmetric distributions. We explore essential properties, including the characterization of the density function for various distribution types, as well as other key aspects such as identifiability, quantiles, stochastic representation, conditional and marginal distributions, moments, Kullback-Leibler Divergence, and parameter estimation.  A Monte Carlo simulation study is performed for examining the performance of the developed parameter estimation method.  Finally, the proposed models are used to analyze socioeconomic data.
\end{abstract}
\smallskip
\noindent
{\small {\bfseries Keywords.} {
Multivariate extended $G$-skew-elliptical distribution $\cdot$ 
\text{EGSE}$_n$ model $\cdot$
Multivariate extended $G$-skew-Student-$t$ $\cdot$
Multivariate extended $G$-skew-normal.}}
\\
{\small{\bfseries Mathematics Subject Classification (2010).} {MSC 60E05 $\cdot$ MSC 62Exx $\cdot$ MSC 62Fxx.}}


\section{Introduction}

Understanding the relationships among multiple jointly observed variables presents a significant challenge in modeling real-world applications.
Data reduction, Grouping, Investigation of the dependence among variables, Prediction, and Hypothesis testing are some of the usual methods. Many of these multivariate methods are based on the multivariate normal distribution. 
There are several applications of multivariate models such as in: 
body composition of athletes \citep{azzalini1996multivariate}; climatology \citep{Genton2010}; outpatient expense and investment in education \citep{svcl:22};
fatigue data \citep{Vila:22}; 
soccer data \citep{Vila:23};
income and consumption data \citep{lima2024assessing}.
We refer the reader to \cite{johnson2002applied} for further details on multivariate analysis.

General families of multivariate distributions have garnered significant attention over the past few decades. 
Bivariate symmetric Heckman models, their mathematical properties, and real data applications were studied by \cite{svcl:22}. 
\cite{Vila:22} extended the definition of univariate log-symmetric distributions to the bivariate case.
\cite{Vila:23} introduced the bivariate unit-log-symmetric model based on the
bivariate log-symmetric distribution.
\cite{fkn:90} extensively presents more general symmetric multivariate models beyond the multivariate normal distribution. In particular, the well-known elliptical symmetric distributions are studied in detail in their book.

However, to better characterize real-world phenomena, studying asymmetric distributions is of great interest.
Furthermore, asymmetry in distributions is common in a wide range of phenomena, including 
the distribution of money 
and the strength of carbon fibers when subjected to tension efforts  \citep[see, for example,][and the references therein]{lima2024assessing,quintino2024estimation}.
Natural extensions of univariate asymmetric models to multivariate ones are widely discussed in the literature. Several authors have made significant advances in the well-known multivariate skew-symmetric and skew-elliptical distributions, which have the multivariate normal distribution as a particular case.
Multivariate versions of the skew-normal distribution were introduced in \cite{azzalini1996multivariate} and \cite{Branco01}.
\cite{arellano2006unified} presented a unified view on skewed distributions arising from selections.
\cite{Genton2010} introduced a family of multivariate log-skew-elliptical distributions, extending several multivariate distributions with positive support.
\cite{Avg:10} introduced
a class of multivariate extended skew-t distributions.

In this paper, we study a new extended family of multivariate skew-elliptical distributions. Our model is based on a multivariate elliptical (symmetric) distribution and in a sequence of real functions $G_1,\ldots, G_n$ appropriately chosen.
In addition, our framework generalizes the multivariate models of \cite{Avg:10} when \( G_i \) are all identity functions, and \cite{Genton2010} when \( G_i \) are all logarithm functions.

Our main contributions are
\begin{itemize}
    \item to derive a new extended family of multivariate skew-elliptical distributions;
    \item to derive analytically several statistical properties of the new distribution;
    \item to propose an estimation procedure for the parameters of the new distribution and validate such procedure via a simulation study and
    \item to apply the proposed models to a real data set on socioeconomic indicators of Switzerland's 47 French-speaking provinces.
\end{itemize}

The paper is organized as follows: in Section \ref{sec:asym_models}, we present a general procedure to construct multivariate asymmetric distributions. Section \ref{Multivariate extended unit} deals with the derivation of the new family of multivariate distributions.
Statistical properties of the new family of distributions are presented in Section \ref{sec:statistical_properties}.
In Section \ref{sec:simulations}, we discuss a simulation study and in Section \ref{sec:applications} the proposed models are applied to a data set on socioeconomic indicators for demonstrating the practical utility of the multivariate asymmetric models introduced here. The last section presents the conclusions.

\section{Multivariate asymmetric distributions}\label{sec:asym_models}

Let $G_1,\ldots,G_n: D\to \mathbb{R}$, $n\in\mathbb{N}$, be a sequence of continuous strictly monotonic functions (which for simplicity of presentation we will assume that they are increasing), where $D\neq\emptyset$ is an arbitrary subset of the set of real numbers. Let $\boldsymbol{X}=(X_1,\ldots,X_n)^\top$ 
denote a $n$-dimensional, absolutely continuous random vector   with support $\mathbb{R}^n$ and let $Z$ be a continuous univariate random variable. Based on $G_1^{-1},\ldots,G_n^{-1}$ (the inverse functions of $G_1,\ldots,G_n$, respectively), $\boldsymbol{X}$ and $Z$, we define a new $n$-dimensional random vector  $\boldsymbol{Y}=(Y_1,\ldots,Y_n)^\top$, with support $D^n$ (the Cartesian product of $n$ sets $D,\ldots,D$), as follows
\begin{align}\label{stochastics-rep}
\boldsymbol{Y}
=
\boldsymbol{T}\,\vert\, \boldsymbol{\lambda}^\top (\boldsymbol{X}-\boldsymbol{\mu})+\tau>Z,
\end{align}
where $\boldsymbol{T}=(T_1,\ldots,T_n)^\top$,  $T_i=G_i^{-1}(X_i), i=1,\ldots,n$, 
$\tau\in\mathbb{R}$ is the extension parameter, $\boldsymbol{\lambda}=(\lambda_1,\ldots,\lambda_n)^\top\in\mathbb{R}^n$ is the skewness parameter vector
and
$\boldsymbol{\mu}=(\mu_1,\ldots,\mu_n)^{\top}\in\mathbb{R}^n$ is a location parameter. 
That is, $\boldsymbol{Y}$ is the conditional random vector for $\boldsymbol{T}$ given  $\boldsymbol{\lambda}^\top (\boldsymbol{X}-\boldsymbol{\mu})+\tau>Z$.


Let $f_{\boldsymbol{Y}}$ be the joint probability density function (PDF) of $\boldsymbol{Y}$.
Bayes' rule provides 
\begin{align}
	f_{\boldsymbol{Y}}(\boldsymbol{y})
	&=\displaystyle
	\dfrac{\displaystyle\int_0^\infty f_{\boldsymbol{T},\boldsymbol{\lambda}^\top (\boldsymbol{X}-\boldsymbol{\mu})-Z+\tau}(\boldsymbol{y},s){\rm d}s}{\mathbb{P}(\boldsymbol{\lambda}^\top (\boldsymbol{X}-\boldsymbol{\mu})+\tau>Z)}, 
	\quad \boldsymbol{y}=(y_1,\ldots,y_n)^\top\in D^n, 
	\nonumber
	\\[0,2cm]
	&= \displaystyle
	f_{\boldsymbol{T}}(\boldsymbol{y}) \,
	\dfrac{\displaystyle\int_0^\infty f_{\boldsymbol{\lambda}^\top (\boldsymbol{X}-\boldsymbol{\mu})-Z+\tau\,\vert\,\boldsymbol{T}=\boldsymbol{y}}(s){\rm d}s}{\mathbb{P}(Z-\boldsymbol{\lambda}^\top (\boldsymbol{X}-\boldsymbol{\mu})<\tau)} 
	\label{pdf-Y}
	\\[0,2cm]
	&=\displaystyle
	f_{\boldsymbol{T}}(\boldsymbol{y}) \,
	{F_Z(\boldsymbol{\lambda}^\top (\boldsymbol{y}_{G}-\boldsymbol{\mu})+\tau\,\vert\, \boldsymbol{X}=\boldsymbol{y}_{G})\over F_{Z-\boldsymbol{\lambda}^\top (\boldsymbol{X}-\boldsymbol{\mu})}(\tau)},  \quad
	\boldsymbol{y}_{G}
	\equiv
	(G_1(y_1),\ldots,G_n(y_n))^\top\in\mathbb{R}^n. \label{pdf-Y-1}
\end{align}

Chain rule gives 
$f_{\boldsymbol{T}}(\boldsymbol{y})
=
f_{\boldsymbol{X}}
(\boldsymbol{y}_{G}) \prod_{i=1}^n G_i'(y_i)$. So, from \eqref{pdf-Y-1}
we have
\begin{align}\label{def-pdf-T}
f_{\boldsymbol{Y}}(\boldsymbol{y})
=
f_{\boldsymbol{X}}(\boldsymbol{y}_{G}) \,
{F_Z(\boldsymbol{\lambda}^\top (\boldsymbol{y}_{G}-\boldsymbol{\mu})+\tau\,\vert\, \boldsymbol{X}=\boldsymbol{y}_{G})\over F_{Z-\boldsymbol{\lambda}^\top (\boldsymbol{X}-\boldsymbol{\mu})}(\tau)}
\, \prod_{i=1}^n G_i'(y_i),
\quad 
\boldsymbol{y}\in D^n,
\end{align}
where $\boldsymbol{y}_{G}$ is as given in \eqref{pdf-Y-1}. 
\begin{remark}
Given the joint distribution of $\boldsymbol{X}$ and $Z$, for each choice of functions $G_1,\ldots,G_n$, $f_{\boldsymbol{Y}}$ represents a large family of asymmetric distributions on the hypercube $D^n$. In this work, for simplicity of presentation, we will assume that $(Z,\boldsymbol{X})^\top$ has  a multivariate elliptical (symmetric) (ELL$_{n+1}$) distribution \citep{fkn:90}; see Section \ref{Multivariate extended unit}.
\end{remark}

Table \ref{table:2-1} presents some examples of functions $G_i$'s for use in \eqref{def-pdf-T}.
\begin{table}[htb!]
	\caption{Some functions $G_i$'s with domain $D$ and its respective inverses and derivatives.}
	\vspace*{0.15cm}
	\centering 
	\begin{tabular}{lclll} 
		\hline
		$G_i(x)$  & $D$ & $G_i^{-1}(x)$ & $G_i'(x)$ & Parameters  
		\\ [0.5ex] 
		\noalign{\hrule height 0.5pt}
		$\tan((x-{1\over 2})\pi)$ & $(0,1)$ & ${1\over 2} + {{\rm arctan}(x)\over \pi}$ & ${\pi\over\sin^2(\pi x)}$  & $-$
		\\ [1ex]
      $- \log(1-x)$ &(0,1)& $1-\exp(-x)$  &$\frac{1}{1-x}$ & $-$
   \\ [1ex] 
      $1- \log(-\log(x))$ &(0,1) &$ \exp(-\exp(-x+1))$ & $\frac{-1}{x \log(x)}$ & $-$
    \\ [1ex] 	
     $ \log(\log(\frac{1}{-x+1})+1)$ & (0,1) &$ 1- \exp(- \exp(x)+1)$ & $ \frac{(-x+1)^{-1}}{ \log( \frac{1}{-x+1})+1}$ & $-$
    \\ [1ex] 
				$\log(\frac{x}{1-x})$ & $(0,1)$ & ${\exp(x)\over1+\exp(x)}$ & ${1\over x(1-x)}$  & $-$
		\\ [1ex]
		$\log(-\log(1-x))$ & $(0,1)$ & $1-\exp(-\exp(x))$ & ${1\over (1-x)\log({1\over 1-x})}$  & $-$
		\\ [1ex]
		$\log(\frac{x^{3}}{1-x^{3}})$& (0,1) & $ \big[\frac{ \exp(x)}{1 + \exp(x)} \big]^{\frac{1}{3}}$ & $\frac{3}{x(1-x^3)}$ & $-$
		\\ [1ex] 
        $ \log(\frac{x^5}{1-x^5})$ & (0,1)& $ \Big[ \frac{ \exp(x)}{1 + \exp(x)} \Big]^{\frac{1}{5}}$ & $\frac{5}{x(1-x^5)}$  & $-$
		\\ [1ex] 	
				$\log(x)$ & $(0,\infty)$ & $\exp(x)$  & ${1\over x}$ & $-$  
		\\ [1ex]
		$x-{1\over x}$ & $(0,\infty)$ & ${1\over 2} (x + \sqrt{x^2 + 4}\,)$ & $1+{1\over x^2}$ & $-$
		\\ [1ex] 
		${1\over \alpha}
		\left(
		\sqrt{x\over \beta}-\sqrt{\beta\over x}
		\right)$ & $(0,\infty)$  & 
		$\beta\left[{\alpha\over 2}x+\sqrt{({\alpha\over 2}x)^2+1}\,\right]^2$
		& ${1\over 2\alpha x}\left(\sqrt{x\over \beta}+\sqrt{\beta\over x}\right)$ & $\alpha,\beta>0$
		\\ [2.5ex] 
		{${2H_i(x)-1\over H_i(x)[1-H_i(x)]}$} & $(0,\infty)$ & { $H^{-1}_i\big({x+\sqrt{x^2+4}\over 2+x+\sqrt{x^2+4}}\big)$}& {${H'_i(x)\over [1-H_i(x)]^2}+{H'_i(x)\over H_i^2(x)}$} & $-$
		\\ [1ex] 
				$ax^p+b$ & $(-\infty,\infty)$ & $({x-b\over a})^{1/p}$  & $apx^{p-1}$ & $a>0, b\in\mathbb{R}$, $p$ odd
	\\ [1ex]
	$\sinh(x)$ & $(-\infty,\infty)$ & $\sinh^{-1}(x)$ & $\cosh(x)$ & $-$
	\\ [1ex] 
	$-\log({1\over F_i(x)}-1)$ & $(-\infty,\infty)$  & $F_i^{-1}({1\over \exp(-x)+1})$ & ${F'_i(x)\over F_i(x)[1-F_i(x)]}$ & $-$
	\\ [1ex] 	
	\hline	
	\end{tabular}
	\label{table:2-1} 
\end{table}
\noindent
	In Table \ref{table:2-1}, $F_i$ (respectively, $H_i$) represents the CDF of a continuous random variable with support on the whole real line (respectively, with positive support). By way of example, we can take $F_i$ as being the CDF of the normal, Gumbel, Student-$t$, logistic, skew normal or symmetric random variable. On the other hand, we can consider $H_i$ as being the CDF of the exponential, Weibull, Gamma, Birnbaum-Saunders (BS) or log-symmetric random variable.

\section{Multivariate extended $G$-skew-elliptical distributions}\label{Multivariate extended unit}

In this section, we provide a formal definition  
of the family of distributions that are the object of study in this work, we refer to the  family of multivariate extended $G$-skew-elliptical (EGSE$_n$) distributions.  In other words, we will obtain the PDF of $\boldsymbol{Y}$ defined in \eqref{stochastics-rep} where $Z$ and $\boldsymbol{X}$ have a probabilistic dependency relationship.

%
Indeed, from now on we assume that the $(n+1)$-dimensional vector $\boldsymbol{V}$, defined as $\boldsymbol{V}=(Z,\boldsymbol{X})^\top$,
%
has  a multivariate elliptical (symmetric) (ELL$_{n+1}$) distribution \citep{fkn:90} with location vector
$\boldsymbol{\mu}_{\boldsymbol{V}}=(0,\boldsymbol{\mu})^\top$, for
$\boldsymbol{\mu}=(\mu_1,\ldots,\mu_n)^{\top}\in\mathbb{R}^n$,  
positive definite $(n+1)\times (n+1)$ dispersion matrix
\begin{align*}
\boldsymbol{\Sigma}_{\boldsymbol{V}}
=
\begin{pmatrix}
1 & \boldsymbol{0}_{n\times 1}^\top \\
\boldsymbol{0}_{n\times 1} & \boldsymbol{\Sigma}
\end{pmatrix},
\quad 
\boldsymbol{\Sigma}
\equiv
\boldsymbol{\Sigma}_{\boldsymbol{X}}=(\sigma_{ij})_{n\times n},
\
\sigma_{ij}={\rm Cov}(X_i,X_j), \ i,j=1,\ldots,n,
\end{align*}
and density generator $g^{(n+1)}$. For simplicity  we use the notation
$\boldsymbol{V}\sim {\rm ELL}_{n+1}(\boldsymbol{\mu}_{\boldsymbol{V}},\boldsymbol{\Sigma}_{\boldsymbol{V}},g^{(n+1)})$.
The density function of $\boldsymbol{V}\sim {\rm ELL}_{n+1}(\boldsymbol{\mu}_{\boldsymbol{V}},\boldsymbol{\Sigma}_{\boldsymbol{V}},g^{(n+1)})$ at $\boldsymbol{x}=(x_1,\ldots,x_{n+1})^\top \in \mathbb{R}^{n+1}$ is given by
\begin{equation}\label{eq:pdf:sym}
f_{\boldsymbol{V}}(\boldsymbol{x})
=
f_{\boldsymbol{V}}(\boldsymbol{x}; \boldsymbol{\mu}_{\boldsymbol{V}},\boldsymbol{\Sigma}_{\boldsymbol{V}},g^{(n+1)}) 
= 
\frac{1}{|\boldsymbol{\Sigma}_{\boldsymbol{V}}|^{1/2} { Z_{g^{(n+1)}}}}\, 
g^{(n+1)} ((\boldsymbol{x} - \boldsymbol{\mu}_{\boldsymbol{V}})^\top \boldsymbol{\Sigma}_{\boldsymbol{V}}^{-1} (\boldsymbol{x} - \boldsymbol{\mu}_{\boldsymbol{V}})), 
\end{equation}
where 
\begin{align*}
Z_{g^{(n+1)}}={\pi^{(n+1)/2}\over \Gamma((n+1)/2)} \, \int_{0}^{\infty} u^{(n+1)/2-1}g^{(n+1)}(u){\rm d}u
\end{align*}
is a normalization constant.

Table \ref{table:1} presents some examples of generators for use in \eqref{eq:pdf:sym}.
\begin{table}[H]
	\caption{Normalization functions $(Z_{g^{(n)}})$ and density generators $(g^{(n)})$.}
	\vspace*{0.15cm}
	\centering 
	\begin{tabular}{llll} 
		\hline
		Multivariate distribution 
		& $Z_{g^{(n)}}$ & $g^{(n)}(x)$ & Parameter 
		\\ [0.5ex] 
		\noalign{\hrule height 0.5pt}
		Extended $G$-skew-Student-$t$
		& ${{\Gamma({\nu/ 2})}(\nu\pi)^{n/2}\over{\Gamma({(\nu+n)/ 2})}}$  
		& $(1+{x\over\nu})^{-(\nu+n)/ 2}$  &  $\nu>0$
		\\ [1ex]		
		Extended $G$-skew-normal
		& $(2\pi)^{n/2}$ & $\exp(-x/2)$ & $-$ 
		\\ [1ex] 	
		\hline	
	\end{tabular}
	\label{table:1} 
\end{table}


%
It is well-known that
all elliptic distributions are invariant to linear transformations \cite[see][]{fkn:90}, that is, if $\boldsymbol{S}\sim {\rm ELL}_n(\boldsymbol{\mu},\boldsymbol{\Omega},g^{(n)})$, for some positive definite dispersion matrix $\boldsymbol{\Omega}$, then $\boldsymbol{c}+\boldsymbol{A}\boldsymbol{S}\sim {\rm ELL}_n(\boldsymbol{c}+\boldsymbol{A}\boldsymbol{\mu},\boldsymbol{A}\boldsymbol{\Omega}\boldsymbol{A}^\top,g^{(n)})$, where $\boldsymbol{A}$ is a square matrix and $\boldsymbol{c}\in\mathbb{R}^n$ is a constant vector.  In particular, this implies that a linear combination of the components of $\boldsymbol{X}$ is again elliptically distributed.  
%
%
%
More precisely, we have
\begin{align}\label{id-1}
Z-\boldsymbol{\lambda}^\top (\boldsymbol{X}-\boldsymbol{\mu})
\sim {\rm ELL}_1
\big(0,
1+\boldsymbol{\lambda}^\top\boldsymbol{\Sigma}\boldsymbol{\lambda}
,g^{(1)}\big).
\end{align}
As a consequence of the last statement, we have that marginals of an elliptic distribution are elliptic. Hence,
\begin{align}\label{id-2}
\boldsymbol{X}\sim {\rm ELL}_n(\boldsymbol{\mu},\boldsymbol{\Sigma},g^{(n)}).
\end{align}

On the other hand,
it is well-known that conditionals of an elliptic distribution are again elliptic \cite[see Theorem 2.18 of][]{fkn:90}. This provides that
\begin{align}\label{id-3}
Z\,\vert\, \boldsymbol{X}=\boldsymbol{x} \sim {\rm ELL}_{1}(0,1,g_{q(\boldsymbol{x})}),
\end{align}
where
\begin{align}\label{def-q}
q(\boldsymbol{x})
=
(\boldsymbol{x}-\boldsymbol{\mu})^\top \boldsymbol{\Sigma}^{-1}(\boldsymbol{x}-\boldsymbol{\mu})
\quad \text{and} \quad 
g_{q(\boldsymbol{x})}(s)
=
{g^{(2)}(s+q(\boldsymbol{x}))\over g^{(1)}(q(\boldsymbol{x})) }.
\end{align}
%
Let $F_{{\rm ELL}_{1}}(\cdot;\, 0,1,g)$ be the CDF corresponding to  ${\rm ELL}_{1}(0,1,g)$ distribution with generator function $g$.
So, from \eqref{id-1}, \eqref{id-2} and \eqref{id-3}, the PDF \eqref{def-pdf-T} of $\boldsymbol{Y}
=
\boldsymbol{T}\,\vert\, \boldsymbol{\lambda}^\top (\boldsymbol{X}-\boldsymbol{\mu})+\tau>Z$ can be written as
\begin{align*}
f_{\boldsymbol{Y}}(\boldsymbol{y})
=
f_{\boldsymbol{X}}(\boldsymbol{y}_{G}) \,
{
	F_{{\rm ELL}_1}(\boldsymbol{\lambda}^\top (\boldsymbol{y}_{G}-\boldsymbol{\mu})+\tau;\, 0,1,g_{q(\boldsymbol{y}_G)})
\over 
F_{{\rm ELL}_1}(\tau;\, 0, 1+\boldsymbol{\lambda}^\top\boldsymbol{\Sigma}\boldsymbol{\lambda},g^{(1)})
}
\, \prod_{i=1}^n G_i'(y_i),
\quad 
\boldsymbol{y}\in D^n,
\end{align*}
with $\boldsymbol{y}_{G}$ being as in \eqref{pdf-Y} and $\boldsymbol{X}\sim {\rm ELL}_n(\boldsymbol{\mu},\boldsymbol{\Sigma},g^{(n)})$.

Note that $F_{{\rm ELL}_1}(\tau=0;\, 0, 1+\boldsymbol{\lambda}^\top\boldsymbol{\Sigma}\boldsymbol{\lambda},g^{(1)})=1/2$ because
$Z-\boldsymbol{\lambda}^\top (\boldsymbol{X}-\boldsymbol{\mu})$ is symmetric about $0$.

\begin{definition}
	We say that a random vector $\boldsymbol{Y}=(Y_1,\ldots,Y_n)^\top$ has a multivariate extended $G$-skew-elliptical (\text{EGSE}$_n$) distribution if $\boldsymbol{Y}$ has PDF given by
\begin{align}\label{def-pdf-T*}
f_{\boldsymbol{Y}}(\boldsymbol{y})
\equiv
f_{\boldsymbol{Y}}(\boldsymbol{y};\boldsymbol{\mu},\boldsymbol{\Sigma},\boldsymbol{\lambda},\tau)
=
f_{\boldsymbol{X}}(\boldsymbol{y}_{G};\boldsymbol{\mu},\boldsymbol{\Sigma}) \,
{
	F_{{\rm ELL}_1}(\boldsymbol{\lambda}^\top (\boldsymbol{y}_{G}-\boldsymbol{\mu})+\tau;\, 0,1,g_{q(\boldsymbol{y}_G)})
	\over 
	F_{{\rm ELL}_1}(\tau;\, 0, 1+\boldsymbol{\lambda}^\top\boldsymbol{\Sigma}\boldsymbol{\lambda},g^{(1)})
}
\, \prod_{i=1}^n G_i'(y_i),
\quad 
\boldsymbol{y}\in D^n,
\end{align}	
	where
	$\boldsymbol{X}\sim \text{ ELL}_n(\boldsymbol{\mu},\boldsymbol{\Sigma},g^{(n)})$.
	For simplicity of notation, we write $\boldsymbol{Y}\sim\text{ EGSE}_n(\boldsymbol{\mu},\boldsymbol{\Sigma},\boldsymbol{\lambda},\tau,g^{(n)})$ and we commonly say that $\boldsymbol{Y}$ is an  EGSE$_n$ random vector.
\end{definition}

\begin{remark}
Standardizing the corresponding random variable of $F_{{\rm ELL}_1}(\cdot;\, 0, 1+\boldsymbol{\lambda}^\top\boldsymbol{\Sigma}\boldsymbol{\lambda},g^{(1)})$, we get 
\begin{align}
F_{{\rm ELL}_1}(\tau;\, 0, 1+\boldsymbol{\lambda}^\top\boldsymbol{\Sigma}\boldsymbol{\lambda},g^{(1)})
&=
F_{{\rm ELL}_1}\left({\tau\over \sqrt{1+\boldsymbol{\lambda}^\top\boldsymbol{\Sigma}\boldsymbol{\lambda}}};\, 0,1,g^{(1)}\right)
\nonumber
\\[0,2cm]
&=
		{1\over Z_{g^{(1)}}}\,
	\int_{-\infty}^{{\tau\over \sqrt{1+\boldsymbol{\lambda}^\top\boldsymbol{\Sigma}\boldsymbol{\lambda}}}} 
	g^{(1)}(s^2)
	{\rm d}s
	\label{id-main0}
	\\[0,2cm]
&=
		{1\over Z_{g^{(1)}}} \,
\int_{-\infty}^{\tau} 
	{1\over \sqrt{1+\boldsymbol{\lambda}^\top\boldsymbol{\Sigma}\boldsymbol{\lambda}}}\, 
g^{(1)}\left({s^2\over 1+\boldsymbol{\lambda}^\top\boldsymbol{\Sigma}\boldsymbol{\lambda}}\right)
{\rm d}s. \label{id-main00}
\end{align}
On the other hand, since $F_{{\rm ELL}_{1}}(\cdot;\, 0,1,g_{q(\boldsymbol{y}_G)})$ is the CDF of  ${\rm ELL}_{1}(0,1,g_{q(\boldsymbol{y}_G)})$ with generator function $g_{q(\boldsymbol{y}_G)}$, as given in \eqref{def-q}, we have
\begin{align}\label{id-main10}
	F_{{\rm ELL}_1}(\boldsymbol{\lambda}^\top (\boldsymbol{y}_{G}-\boldsymbol{\mu})+\tau;\, 0,1,g_{q(\boldsymbol{y}_G)})
	=
	{1\over Z_{g_{q(\boldsymbol{y}_G)}}}
	\int_{-\infty}^{\boldsymbol{\lambda}^\top (\boldsymbol{y}_{G}-\boldsymbol{\mu})+\tau}
	{{g^{(2)}(s^2+q(\boldsymbol{y}_G))}\over g^{(1)}(q(\boldsymbol{y}_G))}
	{\rm d}s,
\end{align}	
where $Z_{g_{q(\boldsymbol{y}_G)}}=\pi\int_{0}^{\infty} g_{q(\boldsymbol{y}_G)}(u){\rm d}u$.
By using \eqref{eq:pdf:sym}, \eqref{id-main0} and \eqref{id-main10} in formula \eqref{def-pdf-T*}, we obtain
\begin{align}\label{pdf-densenv}
	f_{\boldsymbol{Y}}(\boldsymbol{y})
	=
	\frac{1}{|\boldsymbol{\Sigma}|^{1/2} { Z_{g^{(n)}}}}\, 
	g^{(n)} ((\boldsymbol{y}_{G} - \boldsymbol{\mu})^\top \boldsymbol{\Sigma}^{-1} (\boldsymbol{y}_{G} - \boldsymbol{\mu})) \,
	\dfrac{\displaystyle
				{1\over Z_{g_{q(\boldsymbol{y}_G)}}}\,
		\int_{-\infty}^{\boldsymbol{\lambda}^\top (\boldsymbol{y}_{G}-\boldsymbol{\mu})+\tau}
		{g^{(2)}(s^2+q(\boldsymbol{y}_G))\over g^{(1)}(q(\boldsymbol{y}_G))}
		{\rm d}s}{	
		\displaystyle
				{1\over Z_{g^{(1)}}}\,
		\int_{-\infty}^{{\tau\over \sqrt{1+\boldsymbol{\lambda}^\top\boldsymbol{\Sigma}\boldsymbol{\lambda}}}} 
		g^{(1)}(s^2)
		{\rm d}s}.
\end{align}
\end{remark}

\bigskip 
Explicit formulas for the PDF of $\boldsymbol{Y}\sim\text{ EGSE}_n(\boldsymbol{\mu},\boldsymbol{\Sigma},\boldsymbol{\lambda},\tau,g^{(n)})$ corresponding to multivariate extended $G$-skew-Student-$t$
and
multivariate extended $G$-skew-normal 
models (see Table \ref{table:3}), are provided in Subsection \ref{Special cases}.

The \text{EGSE}$_n$ distribution provides a very flexible class of statistical models. Depending on the choice of the functions $G_1,\ldots,G_n$ we have a family of multivariate extended  distributions with presence of asymmetry.
For example, 
for $\boldsymbol{\lambda}=\boldsymbol{0}$, $\tau=0$, $G_1(x)=G_2(x)=\log(-\log(1-x))$, $x\in D=(0,1)$, and $n=2$, we obtain the bivariate unit model studied in reference \cite{Vila:23},
for $\tau=0$ and $G_i(x)=x$, $x\in D=(-\infty,\infty)$, $i=1,\ldots,n$, we obtain the general class of multivariate
skew-elliptical distributions of \cite{Branco01}, and
for $\tau=0$ and $G_i(x)=\log(x)$, $x\in D=(0,\infty)$, $i=1,\ldots,n$, we obtain the multivariate log-skew-elliptical model studied in \cite{Genton2010}. 
In general, for the EGSE$_n$ model, it is not necessary to consider all $G_i$'s equal as in \cite{Vila:23} and \cite{Genton2010}. 
For $g^{(n)}(x)=(1+{x/\nu})^{-(\nu+n)/ 2}$, $\nu>0$, we get the multivariate extended $G$-skew-Student-$t$, which reduces to the
multivariate extended $G$-skew-Cauchy
and multivariate extended $G$-skew-normal
distributions by letting $\nu=1$ and $\nu\to\infty$, respectively.


\section{Statistical properties}\label{sec:statistical_properties}

In this section, we present some special cases of multivariate \text{EGSE}$_n$ PDFs \eqref{def-pdf-T*} and its statistical properties such as reparameterization for to enforce identifiability, invariance properties, stochastic representations, marginal quantiles, conditional and marginal
distributions, closed-forms for the expected value of a function, marginal moments, cross-moments, existence of marginal moments when $D=(0,\infty)$, 
and Kullback-Leibler Divergence, as well as inferential properties.

\subsection{Special cases}\label{Special cases}
In this subsection, we develop some examples of multivariate \text{EGSE}$_n$ PDFs as special cases.

\begin{proposition}[Multivariate extended $G$-skew-Student-$t$]\label{skew-Student}
	Let $g^{(n)}(x)=(1+{x/\nu})^{-(\nu+n)/ 2}$, $x\in\mathbb{R}$, be the PDF generator of the multivariate Student-$t$ distribution with $\nu>0$ degrees of freedom. Then, the PDF of $\boldsymbol{Y}\sim\text{ EGSE}_n(\boldsymbol{\mu},\boldsymbol{\Sigma},\boldsymbol{\lambda},\tau,g^{(n)})$ is given by
	\begin{align}\label{def-pdf-T**}
	f_{\boldsymbol{Y}}(\boldsymbol{y})
	=
	t_n(\boldsymbol{y}_{G};\,\boldsymbol{\mu},\boldsymbol{\Sigma},\nu) \,
	{
		F_{\nu+1}\left( [\boldsymbol{\lambda}^\top (\boldsymbol{y}_{G}-\boldsymbol{\mu})+\tau]
		\sqrt{{\nu+1\over\nu+q(\boldsymbol{y}_G)}}\,\right)
		\over 
		F_{\nu}\Big(
		{\tau
			\over \sqrt{1+\boldsymbol{\lambda}^\top\boldsymbol{\Sigma}\boldsymbol{\lambda}}
		}
		\Big)
	}
	\, \prod_{i=1}^n G_i'(y_i),
	\quad \boldsymbol{y}\in D^n,
	\end{align}
	where $\boldsymbol{y}_{G}$ and $q(\boldsymbol{y}_{G})$  are as given in \eqref{pdf-Y} and \eqref{def-q}, respectively. Moreover,
	$
	t_n(\boldsymbol{y}_{G};\,\boldsymbol{\mu},\boldsymbol{\Sigma},\nu)
	=
	g^{(n)}(q(\boldsymbol{y}_{G}))/(|\boldsymbol{\Sigma}|^{1/2} Z_{g^{(n)}})
	$,
	with $Z_{g^{(n)}}$ being as in Table \ref{table:1},
	denotes the PDF of the usual $n$-dimensional Student-$t$ distribution
	with location $\boldsymbol{\mu}\in\mathbb{R}^n$, positive definite $n\times n$ dispersion matrix $\boldsymbol{\Sigma}$, 
	and degrees of freedom $\nu > 0$, 
	and $F_\nu$ denotes the univariate standard
	Student-$t$ CDF with degrees of freedom $\nu > 0$.
\end{proposition}
\begin{proof}
	By using formula in \eqref{def-pdf-T*},  it is enough to verify that
	\begin{align}\label{ide-1}
	F_{{\rm ELL}_1}(\boldsymbol{\lambda}^\top (\boldsymbol{y}_{G}-\boldsymbol{\mu})+\tau;\, 0,1,g_{q(\boldsymbol{y}_G)})
	=
	F_{\nu+1}\left( [\boldsymbol{\lambda}^\top (\boldsymbol{y}_{G}-\boldsymbol{\mu})+\tau]
	\sqrt{{\nu+1\over\nu+q(\boldsymbol{y}_G)}}\,\right)
	\end{align}
	and
	\begin{align}\label{ide-2}
	F_{{\rm ELL}_1}(\tau;\, 0, 1+\boldsymbol{\lambda}^\top\boldsymbol{\Sigma}\boldsymbol{\lambda},g^{(1)})
	=
	F_{\nu}\left(
	{\tau
		\over \sqrt{1+\boldsymbol{\lambda}^\top\boldsymbol{\Sigma}\boldsymbol{\lambda}}
	}
	\right).
	\end{align}
	The identity \eqref{ide-2} follows directly from identity \eqref{id-main0}. Therefore, it remains to verify \eqref{ide-1}. Indeed, by using identity \eqref{id-main10} and by simple algebraic manipulations, we have
	\begin{align*}
	F_{{\rm ELL}_1}(
	x;\, 0,1,g_{q(\boldsymbol{y}_G)})
	&=
	{1\over Z_{g_{q(\boldsymbol{y}_G)}}}
	\int_{-\infty}^{x}
	{{g^{(2)}(s^2+q(\boldsymbol{y}_G))}\over g^{(1)}(q(\boldsymbol{y}_G))}
	{\rm d}s,
\\[0,2cm]
&=
	{1\over Z_{g_{q(\boldsymbol{y}_G)}}}
	\int_{-\infty}^{x}
	{
		(1+{s^2+q(\boldsymbol{y}_G)\over\nu})^{-(\nu+2)/ 2}
		\over 
		(1+{q(\boldsymbol{y}_G)\over\nu})^{-(\nu+1)/ 2}
	}
	{\rm d}s
	\\[0,2cm]
	&=
	{1\over Z_{g_{q(\boldsymbol{y}_G)}}}
	\int_{-\infty}^{x}
	{
		\left(1+{1\over\nu+1} \left[s\,\sqrt{{\nu+1\over \nu+{q(\boldsymbol{y}_G)}}}\right]^2\right)^{-(\nu+2)/ 2}
		\over 
		\sqrt{1+{q(\boldsymbol{y}_G)\over\nu}}
	}
	{\rm d}s.
	\end{align*}
	By making the change of variable $t=s \sqrt{{(\nu+1)/(\nu+{q(\boldsymbol{y}_G)})}}$, the above identities are briefly written as 
	\begin{align}\label{int-pre}
			F_{{\rm ELL}_1}(
		x;\, 0,1,g_{q(\boldsymbol{y}_G)})
	=
	{1\over Z_{g_{q(\boldsymbol{y}_G)}}}
	\sqrt{{\nu\over \nu+1}}
	\int_{-\infty}^{x\,
		\sqrt{{\nu+1\over \nu+{q(\boldsymbol{y}_G)}}}}
	{
		\left(1+{t^2\over\nu+1}\right)^{-(\nu+2)/ 2}
	}
	{\rm d}t.
	\end{align}
Letting $x\to\infty$ in \eqref{int-pre} we get
\begin{align*}
{1\over Z_{g_{q(\boldsymbol{y}_G)}}} 	\sqrt{{\nu\over \nu+1}} Z_{g^{(1)}_{\nu+1}}
=
			F_{{\rm ELL}_1}(
\infty;\, 0,1,g_{q(\boldsymbol{y}_G)})
=
1,
\end{align*}
where $Z_{g^{(1)}_{\nu+1}}\equiv	\int_{-\infty}^{\infty}
{
	\left(1+{t^2/(\nu+1)}\right)^{-(\nu+2)/ 2}
}
{\rm d}t$ denotes the normalization constant of a student-$t$ distribution with $\nu+1$ degrees of freedom. That is,
	\begin{align}\label{const-new}
	{1\over Z_{g_{q(\boldsymbol{y}_G)}}} 	\sqrt{{\nu\over \nu+1}}
	=
	{1\over Z_{g^{(1)}_{\nu+1}}}
	=
	\left[
	{
	 \Gamma({(\nu+1)/ 2}) 	((\nu+1)\pi)^{1/2} 	
		\over	
		\Gamma({(\nu+2)/ 2})
	}
	\right]^{-1}.
	\end{align}
	
	So, from \eqref{int-pre} and \eqref{const-new}, we have 
	\begin{align*}
					F_{{\rm ELL}_1}(
			\boldsymbol{\lambda}^\top (\boldsymbol{y}_{G}-\boldsymbol{\mu})+\tau
		;\, 0,1,g_{q(\boldsymbol{y}_G)})
	&=
	{1\over Z_{g^{(1)}_{\nu+1}}}
	\int_{-\infty}^{[\boldsymbol{\lambda}^\top (\boldsymbol{y}_{G}-\boldsymbol{\mu})+\tau]
		\sqrt{{\nu+1\over \nu+{q(\boldsymbol{y}_G)}}}}
	{
		\left(1+{t^2\over\nu+1}\right)^{-(\nu+2)/ 2}
	}
	{\rm d}t
	\\[0,2cm]
	&=
	F_{\nu+1}\left( [\boldsymbol{\lambda}^\top (\boldsymbol{y}_{G}-\boldsymbol{\mu})+\tau]
	\sqrt{{\nu+1\over\nu+q(\boldsymbol{y}_G)}}\,\right).
	\end{align*}
	Then, the required formula in \eqref{ide-1} follows.
\end{proof}

By letting $\nu\to\infty$ in Proposition \ref{skew-Student}, the following result follows.
\begin{proposition}[Multivariate extended $G$-skew-normal]\label{prop-3}
	Let $\boldsymbol{Y}\sim\text{ EGSE}_n(\boldsymbol{\mu},\boldsymbol{\Sigma},\boldsymbol{\lambda},\tau,g^{(n)})$, where $g^{(n)}(x)=\exp(-x/2)$, $x\in\mathbb{R}$, is the PDF generator of the multivariate Gaussian distribution. Then, the PDF of $\boldsymbol{Y}$ at $\boldsymbol{y}\in D^n$ is given by
	\begin{align}\label{def-pdf-T****}
	f_{\boldsymbol{Y}}(\boldsymbol{y})
	=
	\phi_n(\boldsymbol{y}_{G};\,\boldsymbol{\mu},\boldsymbol{\Sigma}) \,
	{
		\Phi\left(\boldsymbol{\lambda}^\top (\boldsymbol{y}_{G}-\boldsymbol{\mu})+\tau\right)
		\over 
		\Phi\Big(
		{\tau
			\over \sqrt{1+\boldsymbol{\lambda}^\top\boldsymbol{\Sigma}\boldsymbol{\lambda}}
		}
		\Big)
	}
	\, \prod_{i=1}^n G_i'(y_i),
	\end{align}
	where $\boldsymbol{y}_{G}$ is as given in \eqref{pdf-Y}. Here,
	$
	\phi_n(\boldsymbol{y}_{G};\,\boldsymbol{\mu},\boldsymbol{\Sigma},\nu)
	=
	g^{(n)}((\boldsymbol{y}_{G}-\boldsymbol{\mu})^\top \boldsymbol{\Sigma}^{-1}(\boldsymbol{y}_{G}-\boldsymbol{\mu}))/(|\boldsymbol{\Sigma}|^{1/2} Z_{g^{(n)}})
	$,
	with $Z_{g^{(n)}}$ being as in Table \ref{table:1},
	denotes the PDF of the usual $n$-dimensional Gaussian distribution
	with location $\boldsymbol{\mu}\in\mathbb{R}^n$ and positive definite $n\times n$ dispersion matrix $\boldsymbol{\Sigma}$, 
	and $\Phi$ denotes the univariate standard
	Gaussian CDF.
\end{proposition}

Table \ref{table:3} summarizes the results found in Propositions \ref{skew-Student} and \ref{prop-3}.
\begin{table}[H]
	\caption{Densities $f_{\boldsymbol{Y}}$ of the EGSE$_n$ distributions of Table \ref{table:1}.}
	\vspace*{0.15cm}
	\centering 
	\begin{tabular}{llll} 
		\hline
		Multivariate distribution 
		& $f_{\boldsymbol{Y}}(\boldsymbol{y})$ 
		\\ [0.5ex] 
		\noalign{\hrule height 0.5pt}
		Extended $G$-skew-Student-$t$
		& $	t_n(\boldsymbol{y}_{G};\,\boldsymbol{\mu},\boldsymbol{\Sigma},\nu) \,
		{
			F_{\nu+1}\left( [\boldsymbol{\lambda}^\top (\boldsymbol{y}_{G}-\boldsymbol{\mu})+\tau]
			\sqrt{{\nu+1\over\nu+q(\boldsymbol{y}_G)}}\,\right)
			\over 
			F_{\nu}\big(
			{\tau
				\over \sqrt{1+\boldsymbol{\lambda}^\top\boldsymbol{\Sigma}\boldsymbol{\lambda}}
			}
			\big)
		}
		\, \prod_{i=1}^n G_i'(y_i)$  
		\\ [3ex]		
		Extended $G$-skew-normal
		& $	\phi_n(\boldsymbol{y}_{G};\,\boldsymbol{\mu},\boldsymbol{\Sigma}) \,
		{
			\Phi\left(\boldsymbol{\lambda}^\top (\boldsymbol{y}_{G}-\boldsymbol{\mu})+\tau\right)
			\over 
			\Phi\big(
			{\tau
				\over \sqrt{1+\boldsymbol{\lambda}^\top\boldsymbol{\Sigma}\boldsymbol{\lambda}}
			}
			\big)
		}
		\, \prod_{i=1}^n G_i'(y_i)$ 
		\\ [1ex] 	
		\hline	
	\end{tabular}
	\label{table:3} 
\end{table}

\subsection{Reparameterization for to enforce identifiability}

In general, identifiability is lost when a multivariate normal distribution is reduced by conditioning \citep{Florens:90}. This leads us to believe that for any choices of density generators $(g^{(n)})$ the EGSE$_n$ model \eqref{def-pdf-T*} loses identifiability.
%
It is natural to ask whether through reparameterization the model gains the property of identifiability. At least for the extended $G$-skew-normal distribution (see Table \ref{table:3}) the answer is positive.
To verify this statement we consider the reparameterization $(\boldsymbol{\mu},\boldsymbol{\Sigma},\boldsymbol{\lambda},\tau)^\top\longmapsto\boldsymbol{\psi}=(\boldsymbol{\mu},\boldsymbol{\Sigma}_*,\boldsymbol{\delta},\boldsymbol{\gamma})^\top$, where
\begin{align}\label{reparametrizing}
	&\boldsymbol{\Sigma}_*\equiv\boldsymbol{\omega}^{-1}\boldsymbol{\Sigma}\boldsymbol{\omega}^{-1}
	=
	\begin{pmatrix}
		1 & {\sigma_{12}\over \sqrt{\sigma_{11}\sigma_{22}}} &\ldots & {\sigma_{1n}\over \sqrt{\sigma_{11}\sigma_{nn}}} \\
		{\sigma_{21}\over \sqrt{\sigma_{22}\sigma_{11}}} & 1 & \cdots & {\sigma_{2n}\over \sqrt{\sigma_{22}\sigma_{nn}}} \\
		\vdots& \vdots& \ddots& \vdots\\
		{\sigma_{n1}\over \sqrt{\sigma_{nn}\sigma_{11}}} & {\sigma_{n2}\over \sqrt{\sigma_{nn}\sigma_{22}}} &\cdots& 1
	\end{pmatrix},
\end{align}
with
\begin{align*}
	&\boldsymbol{\omega}\equiv\sqrt{{\rm diag}(\boldsymbol{\Sigma})}
	=
	\begin{pmatrix}
		\sqrt{\sigma_{11}} & 0 &\ldots & 0\\
		0&\sqrt{\sigma_{22}}& \cdots & 0 \\
		\vdots& \vdots& \ddots& \vdots\\
		0&0&\cdots&\sqrt{\sigma_{nn}}
	\end{pmatrix},
	\nonumber
\end{align*}
is the correlation matrix 
and
\begin{align}\label{initial-param}
		\boldsymbol{\delta}\equiv{\boldsymbol{\Sigma}_*\boldsymbol{\lambda}\over \sqrt{1+\boldsymbol{\lambda}^\top\boldsymbol{\Sigma}_*\boldsymbol{\lambda}}},
		\quad 	
	\gamma\equiv{\tau\over \sqrt{1+\boldsymbol{\lambda}^\top\boldsymbol{\Sigma}_*\boldsymbol{\lambda}}}.
\end{align}

In what remains of this subsection we will prove that the parametrization $\boldsymbol{\psi}$ is identifiable. Indeed, note that 
\begin{align}\label{id-implication}
	\boldsymbol{\delta}^\top
	=
	{\boldsymbol{\lambda}^\top\boldsymbol{\Sigma}_*\over \sqrt{1+\boldsymbol{\lambda}^\top\boldsymbol{\Sigma}_*\boldsymbol{\lambda}}}
	\quad \Longrightarrow \quad 
	\sqrt{1+\boldsymbol{\lambda}^\top\boldsymbol{\Sigma}_*\boldsymbol{\lambda}}
		=
		{
		1\over \sqrt{1-\boldsymbol{\delta}^\top\boldsymbol{\Sigma}_*^{-1}
		\delta}
		}.
\end{align}
By using \eqref{id-implication}, we obtain
\begin{itemize}
	\item 
	\begin{align}\label{lambda-delta}
	\boldsymbol{\lambda}^\top
	=
	\boldsymbol{\delta}^\top\boldsymbol{\Sigma}_*^{-1} \sqrt{1+\boldsymbol{\lambda}^\top\boldsymbol{\Sigma}_*\boldsymbol{\lambda}}
	=
	{\boldsymbol{\delta}^\top\boldsymbol{\Sigma}_*^{-1}\over \sqrt{1-\boldsymbol{\delta}^\top\boldsymbol{\Sigma}_*^{-1}
			\delta}},
	\end{align}
	\item 
	\begin{align}\label{tau-delta}
	\tau=\gamma\sqrt{1+\boldsymbol{\lambda}^\top\boldsymbol{\Sigma}_*\boldsymbol{\lambda}}
	=
	{\gamma\over \sqrt{1-\boldsymbol{\delta}^\top\boldsymbol{\Sigma}_*^{-1}
			\delta}}.
	\end{align}
\end{itemize}
Hence, by \eqref{initial-param}, \eqref{lambda-delta} and \eqref{tau-delta},
the extended $G$-skew-normal PDF (see Table \ref{table:3}) can be written as a function of $\boldsymbol{\psi}$ as follows:
\begin{align}\label{id-skew-gen}
	f_{\boldsymbol{Y}}(\boldsymbol{y};\boldsymbol{\psi})
	=
	\phi_n(\boldsymbol{y}_{G};\,\boldsymbol{\mu},\boldsymbol{\Sigma}_*) \,
	{
		\Phi\left(\displaystyle
		{\boldsymbol{\delta}^\top\boldsymbol{\Sigma}_*^{-1}
			(\boldsymbol{y}_{G}-\boldsymbol{\mu})
			+
			\gamma
			\over \sqrt{1-\boldsymbol{\delta}^\top\boldsymbol{\Sigma}_*^{-1}
				\delta}
			}
			\right)
		\over 
		\Phi(\gamma)
	}
	\, \prod_{i=1}^n G_i'(y_i)
	=
		f_{\rm SN}(\boldsymbol{y}_G;\boldsymbol{\psi})\, \prod_{i=1}^n G_i'(y_i),
\end{align}
where $f_{\rm SN}(\cdot;\boldsymbol{\psi})$ is the skew-normal distribution defined as \citep[see][]{Castro13}
\begin{align}\label{s-n-distribution}
	f_{\rm SN}(\boldsymbol{z};\boldsymbol{\psi})
	\equiv
	\phi_n(\boldsymbol{z};\,\boldsymbol{\mu},\boldsymbol{\Sigma}_*) \,
	{
		\Phi\left(\displaystyle
		{\boldsymbol{\delta}^\top\boldsymbol{\Sigma}_*^{-1}
			(\boldsymbol{z}-\boldsymbol{\mu})
			+
			\gamma
			\over \sqrt{1-\boldsymbol{\delta}^\top\boldsymbol{\Sigma}_*^{-1}
				\delta}
		}
		\right)
		\over 
		\Phi(\gamma)
	},
	\quad \boldsymbol{z}\in\mathbb{R}^n,
\end{align}

By using the $r$th cumulants of random vector corresponding to PDF $f_{\rm SN}(\cdot;\boldsymbol{\psi})$, in Section 2 of \cite{Castro13}, it was proven that the skew-normal distribution  \eqref{s-n-distribution}
is identifiable. In other words, it was shown that
\begin{align*}
	f_{\rm SN}(\boldsymbol{z};\boldsymbol{\psi})
	=
	f_{\rm SN}(\boldsymbol{z};\boldsymbol{\psi}'), \ \forall \boldsymbol{z}\in\mathbb{R}^n
	\quad \Longrightarrow \quad 
	\boldsymbol{\psi}=\boldsymbol{\psi}'.
\end{align*}
As an immediate consequence of the above result, we obtain
\begin{align*}
		f_{\boldsymbol{Y}}(\boldsymbol{y};\boldsymbol{\psi})
		\stackrel{\eqref{id-skew-gen}}{=}
		f_{\rm SN}(\boldsymbol{y}_G;\boldsymbol{\psi})\, \prod_{i=1}^n G_i'(y_i)
		=
		f_{\rm SN}(\boldsymbol{y}_G;\boldsymbol{\psi}')\, \prod_{i=1}^n G_i'(y_i)
		\stackrel{\eqref{id-skew-gen}}{=}
		f_{\boldsymbol{Y}}(\boldsymbol{y};\boldsymbol{\psi}')
			, \ \forall \boldsymbol{y}\in D^n
			\quad \Longrightarrow \quad 
			\boldsymbol{\psi}=\boldsymbol{\psi}'.
\end{align*}
This shows the identifiability of the  extended $G$-skew-normal distribution model when considering reparameterization $\boldsymbol{\psi}=(\boldsymbol{\mu},\boldsymbol{\Sigma}_*,\boldsymbol{\delta},\boldsymbol{\gamma})^\top$.

\subsection{Invariance properties}

In this subsection, we show that 
for any even function  $\vartheta:D^n\to\mathbb{R}$, i.e. a function such that  $\vartheta(-\boldsymbol{y})=\vartheta(\boldsymbol{y})$, $\boldsymbol{y}\in D^n$, and for any odd functions $G_1,\ldots,G_n$, i.e. functions such that  $G_1(-y)=-G_1(y),\ldots,G_n(-y)=-G_n(y)$, $y\in D$,
the joint distribution of the function $\vartheta(\boldsymbol{Y})$ does
not depend on the skewness parameter $\boldsymbol{\lambda}$, for an
\text{EGSE}$_n$ random vector $\boldsymbol{Y}$ centered at $\boldsymbol{\mu}=\boldsymbol{0}$ and with extension parameter $\tau=0$.

\begin{proposition}\label{invariance}
	If $\boldsymbol{Y}\sim\text{ EGSE}_n(\boldsymbol{0},\boldsymbol{\Sigma},\boldsymbol{\lambda},0,g^{(n)})$, then the distribution of $\vartheta(\boldsymbol{Y})$, where $\vartheta$ is an even function and $G_1,\ldots,G_n$ are odd functions, does not depend on the function $F_{{\rm ELL}_1}$.
\end{proposition}
\begin{proof}
	The proof of this result follows the same reasoning as the proof of Proposition 3.1 in \cite{gl:05}. For completeness and for the reader's convenience, we present the proof here.
	
	If we show that the characteristic function of $\vartheta(\boldsymbol{Y})$, denoted by $\phi_{ \vartheta(\boldsymbol{Y})}(t)=\mathbb{E}[\exp(it \vartheta(\boldsymbol{Y}))]$, $t\in\mathbb{R}$, does not depend on the function $F_{{\rm ELL}_1}$, the proof ends.
	Indeed, note that $\phi_{ \vartheta(\boldsymbol{Y})}(t)$ can be  written as
	\begin{align}
	\phi_{ \vartheta(\boldsymbol{Y})}(t)
	&=
	2
	\int_{A^-}
	\exp(it \vartheta(\boldsymbol{y}))
	f_{\boldsymbol{X}}(\boldsymbol{y}_{G}) \,
	{
		F_{{\rm ELL}_1}(\boldsymbol{\lambda}^\top \boldsymbol{y}_{G};\, 0,1,g_{q(\boldsymbol{y}_G)})
	}
	\, \prod_{i=1}^n G_i'(y_i)
	{\rm d} 
	\boldsymbol{y}
	\nonumber
	\\
	&+2
	\int_{A^+}
	\exp(it \vartheta(\boldsymbol{y}))
	f_{\boldsymbol{X}}(\boldsymbol{y}_{G}) \,
	{
		F_{{\rm ELL}_1}(\boldsymbol{\lambda}^\top \boldsymbol{y}_{G};\, 0,1,g_{q(\boldsymbol{y}_G)})
	}
	\, \prod_{i=1}^n G_i'(y_i)
	{\rm d} 
	\boldsymbol{y},
	\label{id-cf-0}
	\end{align}
	where $\boldsymbol{y}_{G}$ is as given in \eqref{pdf-Y}, $A^+=\{(y_1,\ldots,y_n)^\top\in D^n:y_1\geqslant 0\}$ and $A^-=\{(y_1,\ldots,y_n)^\top\in D^n:y_1< 0\}$.
	
	Moreover, using the facts that $\vartheta$ is an even function, $G_1,\ldots,G_n$ are odd functions and that $F_{{\rm ELL}_1}$  is a skewing function, i.e. 	$F_{{\rm ELL}_1}(\boldsymbol{\lambda}^\top (G_1(-y_1),\ldots,G_n(-y_n))^\top;\, 0,1,g_{q(\boldsymbol{y}_G)})=1-F_{{\rm ELL}_1}(\boldsymbol{\lambda}^\top \boldsymbol{y}_{G};\, 0,1,g_{q(\boldsymbol{y}_G)})$, we have
	\begin{align}
	2
	\int_{A^-}
	\exp(it \vartheta(\boldsymbol{y}))
	f_{\boldsymbol{X}}(\boldsymbol{y}_{G}) \,
	&
	{
		F_{{\rm ELL}_1}(\boldsymbol{\lambda}^\top \boldsymbol{y}_{G};\, 0,1,g_{q(\boldsymbol{y}_G)})
	}
	\, \prod_{i=1}^n G_i'(y_i)
	{\rm d} 
	\boldsymbol{y}
	\nonumber
	\\[0,2cm]
	&=
	2
	\int_{A^+}
	\exp(it \vartheta(-\boldsymbol{y}))
	f_{\boldsymbol{X}}(G_1(-y_1),\ldots,G_n(-y_n)) \,
		\nonumber
	\\[0,2cm]
	&\times {
		F_{{\rm ELL}_1}(\boldsymbol{\lambda}^\top (G_1(-y_1),\ldots,G_n(-y_n))^\top;\, 0,1,g_{q(\boldsymbol{y}_G)})
	}
	\, \prod_{i=1}^n G_i'(-y_i)
	{\rm d} 
	\boldsymbol{y}
	\nonumber
	\\[0,2cm]
	&=
	2
	\int_{A^+}
	\exp(it \vartheta(\boldsymbol{y}))
	f_{\boldsymbol{X}}(\boldsymbol{y}_{G})
	\, \prod_{i=1}^n G_i'(y_i)
	{\rm d} 
	\boldsymbol{y}
		\nonumber
	\\[0,2cm]
	&
	-
	2
	\int_{A^+}
	\exp(it \vartheta(\boldsymbol{y}))
	f_{\boldsymbol{X}}(\boldsymbol{y}_{G}) \,
	{
		F_{{\rm ELL}_1}(\boldsymbol{\lambda}^\top \boldsymbol{y}_{G};\, 0,1,g_{q(\boldsymbol{y}_G)})
	}
	\, \prod_{i=1}^n G_i'(y_i)
	{\rm d} 
	\boldsymbol{y},
	\label{id-cf}
	\end{align}
	where in the last equality we used the well-known fact that the derivative of an odd function is even.
	
	By combining \eqref{id-cf-0} and \eqref{id-cf}, we get
	\begin{align*}
	\phi_{ \vartheta(\boldsymbol{Y})}(t)
	=
	2
	\int_{A^+}
	\exp(it \vartheta(\boldsymbol{y}))
	f_{\boldsymbol{X}}(\boldsymbol{y}_{G})
	\, \prod_{i=1}^n G_i'(y_i)
	{\rm d} 
	\boldsymbol{y}.
	\end{align*}
	In other words, we have proven that the distribution of $\vartheta(\boldsymbol{Y})$ does not depend on the function $F_{{\rm ELL}_1}$, thus completing the proof.
\end{proof}

\begin{remark}
	Some examples of odd functions $G_i$'s with support on the real line that we can consider in Proposition \ref{invariance} are
	$G_i(x)=ax^p+b$, with $a>0, b=0$, $p$ odd, or
	$G_i(x)=\sinh(x)$ (see Table \ref{table:2-1}).
\end{remark}

Applying Proposition \ref{invariance} we immediately have the following two results.
\begin{corollary}
	If $\boldsymbol{Y}\sim\text{ EGSE}_n(\boldsymbol{0},\boldsymbol{\Sigma},\boldsymbol{\lambda},0,g^{(n)})$, then
	the distribution of $\boldsymbol{Y}\boldsymbol{Y}^\top$ does not
	depend on the function $F_{{\rm ELL}_1}$.
\end{corollary}

\begin{corollary}
	Let $A_1,\ldots, A_m$ be $n\times n$ real matrices and let $\boldsymbol{Y}\sim\text{ EGSE}_n(\boldsymbol{0},\boldsymbol{\Sigma},\boldsymbol{\lambda},0,g^{(n)})$. Then the joint distribution of the quadratic forms $(\boldsymbol{Y}A_1\boldsymbol{Y}^\top,\ldots, \boldsymbol{Y}A_m\boldsymbol{Y}^\top)^\top$ does not
	depend on the function $F_{{\rm ELL}_1}$.
\end{corollary}

\subsection{Stochastic representation}\label{Stochastic representation}

Let
$
\boldsymbol{W}
=
(W_{1},\ldots, W_{n})^\top 
=
\boldsymbol{X}\,\vert\, \boldsymbol{\lambda}^\top (\boldsymbol{X}-\boldsymbol{\mu})+\tau>Z
$, 
where  
$\boldsymbol{V}=(Z,\boldsymbol{X})^\top\sim {\rm ELL}_{n+1}(\boldsymbol{\mu}_{\boldsymbol{V}},\boldsymbol{\Sigma}_{\boldsymbol{V}},g^{(n+1)})$, and $\boldsymbol{\mu}_{\boldsymbol{V}}$ and $\boldsymbol{\Sigma}_{\boldsymbol{V}}$ as defined in \eqref{eq:pdf:sym}.
Using the same steps to obtain the density of $\boldsymbol{Y}$ in \eqref{def-pdf-T*}, it can be seen that the PDF of $\boldsymbol{W}$ is given by
\begin{align}\label{def-pdf-W}
f_{\boldsymbol{W}}(\boldsymbol{w})
=
f_{\boldsymbol{X}}(\boldsymbol{w}) \,
{
	F_{{\rm ELL}_1}(\boldsymbol{\lambda}^\top (\boldsymbol{w}-\boldsymbol{\mu})+\tau;\, 0,1,g_{q(\boldsymbol{w})})
	\over 
	F_{{\rm ELL}_1}(\tau;\, 0, 1+\boldsymbol{\lambda}^\top\boldsymbol{\Sigma}\boldsymbol{\lambda},g^{(1)})
},
\quad 
\boldsymbol{w}\in\mathbb{R}^n.
\end{align}
A random vector  $\boldsymbol{W}$ with density given by \eqref{def-pdf-W} is said to have a multivariate extended skew-elliptical (\text{ESE}$_n$) distribution. For simplicity, we write $\boldsymbol{W}\sim\text{ ESE}_n(\boldsymbol{\mu},\boldsymbol{\Sigma},\boldsymbol{\lambda},\tau,g^{(n)})$.	

Table \ref{table:4} presents some examples of density functions for $\boldsymbol{W}$.
\begin{table}[H]
	\caption{Some particular densities for the \text{ESE}$_n$ random vector.}
	\vspace*{0.15cm}
	\centering 
	\begin{tabular}{llll} 
		\hline
		Multivariate distribution 
		& $f_{\boldsymbol{W}}(\boldsymbol{w})$ 
		\\ [0.5ex] 
		\noalign{\hrule height 0.5pt}
		Extended skew-Student-$t$
		& $	t_n(\boldsymbol{w};\,\boldsymbol{\mu},\boldsymbol{\Sigma},\nu) \,
		{
			F_{\nu+1}\left( [\boldsymbol{\lambda}^\top (\boldsymbol{w}-\boldsymbol{\mu})+\tau]
			\sqrt{{\nu+1\over\nu+q(\boldsymbol{w})}}\,\right)
			\over 
			F_{\nu}\big(
			{\tau
				\over \sqrt{1+\boldsymbol{\lambda}^\top\boldsymbol{\Sigma}\boldsymbol{\lambda}}
			}
			\big)
		}$  
		\\ [3ex]		
		Extended skew-normal
		& $	\phi_n(\boldsymbol{w};\,\boldsymbol{\mu},\boldsymbol{\Sigma}) \,
		{
			\Phi\left(\boldsymbol{\lambda}^\top (\boldsymbol{w}-\boldsymbol{\mu})+\tau\right)
			\over 
			\Phi\big(
			{\tau
				\over \sqrt{1+\boldsymbol{\lambda}^\top\boldsymbol{\Sigma}\boldsymbol{\lambda}}
			}
			\big)
		}$ 
		\\ [1ex] 	
		\hline	
	\end{tabular}
	\label{table:4} 
\end{table}

Let $\boldsymbol{Y}=(Y_1,\ldots,Y_n)^\top\sim\text{ EGSE}_n(\boldsymbol{\mu},\boldsymbol{\Sigma},\boldsymbol{\lambda},\tau,g^{(n)})$.
From \eqref{stochastics-rep}, 
$\boldsymbol{Y}
=
\boldsymbol{T}\,\vert\, \boldsymbol{\lambda}^\top (\boldsymbol{X}-\boldsymbol{\mu})+\tau>Z$, with  $\boldsymbol{T}=(G_1^{-1}(X_1),\ldots,G_n^{-1}(X_n))^\top$ and $(Z,\boldsymbol{X})^\top$ as defined in \eqref{def-pdf-W}.
Then, it is clear that their joint distribution can be written as
\begin{align}\label{identity-CDFs}
		\mathbb{P}(Y_1\leqslant y_1, \ldots,Y_n\leqslant y_n)
		&=
		\mathbb{P}(G_1^{-1}(X_1)\leqslant y_1,\ldots,G_n^{-1}(X_n)\leqslant y_n\,\vert\, \boldsymbol{\lambda}^\top (\boldsymbol{X}-\boldsymbol{\mu})+\tau>Z)
		\nonumber
		\\[0,2cm]
		&=
		\mathbb{P}(G_1^{-1}(W_1)\leqslant y_1,\ldots,G_n^{-1}(W_n)\leqslant y_n),
		\quad 
		\forall (y_1,\ldots,y_n).	
\end{align}
That is,
\begin{align}\label{rep-stoch}
	\boldsymbol{Y}
	=
(Y_1,\ldots,Y_n)^\top
\stackrel{d}{=}
(G_1^{-1}(W_1),\ldots,G_n^{-1}(W_n))^\top,
\end{align}
with $\stackrel{d}{=}$ being equality in distribution.

Letting $y_k\to \infty$ in \eqref{identity-CDFs}, in all $y_k$ except the $i$th component, we obtain
\begin{align*}
	\mathbb{P}(Y_i\leqslant y_i)
	=
	\mathbb{P}(G_i^{-1}(W_i)\leqslant y_i),
	\quad \forall i=1,\ldots,n.
\end{align*}
In other words,
\begin{align}\label{identity-CDFs-univariate}
Y_i\stackrel{d}{=}G_i^{-1}(W_i),
\quad \forall i=1,\ldots,n.
\end{align}

\subsection{Marginal quantiles}

Given $p\in(0,1)$, the marginal $p$-quantile of $\boldsymbol{Y}=(Y_1,\ldots,Y_n)^\top\sim\text{ EGSE}_n(\boldsymbol{\mu},\boldsymbol{\Sigma},\boldsymbol{\lambda},\tau,g^{(n)})$ will be denoted by $Q_{Y_i}(p)$. So, from  \eqref{identity-CDFs-univariate} we have
\begin{align*}
	p=\mathbb{P}(Y_i\leqslant Q_{Y_i}(p))
	=
	\mathbb{P}(G_i^{-1}(W_i)\leqslant Q_{Y_i}(p))
	=
	\mathbb{P}(W_i\leqslant G_i(Q_{Y_i}(p))),
	\quad i=1,\ldots,n,
\end{align*}
with $\boldsymbol{W}=(W_1,\ldots,W_n)^\top\sim\text{ ESE}_n(\boldsymbol{\mu},\boldsymbol{\Sigma},\boldsymbol{\lambda},\tau,g^{(n)})$.
Equivalently,
\begin{align*}
Q_{W_i}(p)
=
G_i(Q_{Y_i}(p))
\end{align*}
if and only if
\begin{align*}
Q_{Y_i}(p)
=
G_i^{-1}(Q_{W_i}(p)), \quad i=1,\ldots,n.
\end{align*}
In other words, if the $p$-quantile of $W_i$ is known, then the $p$-quantile of $Y_i$ can be determined explicitly.

\subsection{Conditional and marginal distributions}
In the context of multivariate sample selection models \citep{heckman76}, the interest lies in finding the PDF of $Y_i\,|\, Y_j > \kappa$, $i\neq j\in\{1,\ldots,n\}$, given that $\boldsymbol{Y}=({Y}_1,\ldots, Y_n)^\top\sim\text{ EGSE}_n(\boldsymbol{\mu},\boldsymbol{\Sigma},\boldsymbol{\lambda},\tau,g^{(n)})$, with $\kappa\in D$. For this purpose, let
 $\boldsymbol{W}=(W_{1},\ldots, W_{n})^\top \sim\text{ ESE}_n(\boldsymbol{\mu},\boldsymbol{\Sigma},\boldsymbol{\lambda},\tau,g^{(n)})$ be a multivariate extended skew-elliptical random vector. From Subsection \ref{Stochastic representation} we know that
$
\boldsymbol{W}
=
\boldsymbol{X}\,\vert\, \boldsymbol{\lambda}^\top (\boldsymbol{X}-\boldsymbol{\mu})+\tau>Z
$.

Analogously to the steps developed in \eqref{pdf-Y},
Bayes' rule provides 
\begin{align}\label{main-identity}
f_{Y_i\,|\,Y_j > \kappa}(y)
=
f_{Y_i}(y) \,
\dfrac{\displaystyle\int_{\kappa}^\infty f_{Y_j\,\vert\,Y_i=y}(s){\rm d}s}{\mathbb{P}(Y_j>\kappa)},
\quad y\in D, \ \kappa\in D.	
\end{align}
If $Y_i=y$ then $W_{i} =G_i(y)$. So, the distribution of $Y_j\,\vert\,Y_i=y$ is the same as the distribution of $G_{j}^{-1}(W_{j})\,\vert\, W_{i} =G_i(y)$. Consequently, the PDF of $Y_j$
given $Y_i=y$ is given by
\begin{align}\label{main-identity-1}
f_{Y_j\,\vert\,Y_i=y}(s)
=
f_{W_{j}\,\vert\,W_{i} =G_i(y)}(G_j(s)) \,G_{j}'(s).
\end{align}
Since, by \eqref{identity-CDFs-univariate}, 
\begin{align}\label{marginal-pdf}
f_{{Y}_i}(y)
=
f_{W_{i}}(G_i(y)) G_{i}'(y)
\quad
\text{and} 
\quad
f_{Y_j}(s)
=
f_{W_{j}}(G_j(s)) G_{j}'(s),
\end{align}  
from \eqref{main-identity} and \eqref{main-identity-1} we get
\begin{align*}
f_{Y_i\,|\,Y_j > \kappa}(y)
=
f_{W_{i}}(G_i(y)) G_{i}'(y) \ 
\dfrac{\displaystyle
	\int_\kappa^\infty 
f_{W_{j}\,\vert\,W_{i} =G_i(y)}(G_j(s)) \,G_{j}'(s){\rm d}s  
}{ 
	\displaystyle\int_\kappa^\infty 
f_{W_{j}}(G_j(s)) G_{j}'(s) {\rm d}s}.
\end{align*}
Equivalently,
\begin{align}\label{main-formula}
f_{Y_i\,|\,Y_j > \kappa}(y)
=
f_{W_{i}}(G_i(y)) G_{i}'(y) \ 
\dfrac{\displaystyle
	S_{W_{j}\,\vert\,W_{i} =
		G_i(y)}(G_j(\kappa))  
}{ 
	S_{W_{j}}(G_j(\kappa))},
\quad y\in D, \ \kappa\in D,	
\end{align}
where $S_X$ denotes the survival function (SF) of $X$.
In other words, to determine the distribution of $Y_i\,|\,Y_j>\kappa$ it is sufficient to know the  unconditional and conditional distributions of the multivariate extended skew-elliptical random vector $\boldsymbol{W}$.

%
In what remains of this subsection
we present closed-forms for the PDFs of $Y_i\,|\,Y_j>\kappa$  and $Y_i$ by considering the Student-$t$
and Gaussian generator densities.

\subsubsection{Student-$t$ density generator}
Let $g^{(n)}(x)=(1+x/\nu)^{-(\nu+n)/2}$, $x\in\mathbb{R}$ (see Table \ref{table:1}), 
be the Student-$t$ density generator of the \text{EGSE}$_n$ (multivariate extended $G$-skew-Student-$t$) distribution. 

\begin{definition}\label{def-main-2}
	A random variable $X$ follows a univariate extended skew-Student-$t$ (EST$_1$) distribution, denoted by $X\sim\text{ EST}_1(\mu,\sigma^2,\lambda,\nu,\tau)$, if its PDF is given by \cite[see][]{Avg:10}
	\begin{align*}
	f_{{\rm EST}_1}(x;\mu,\sigma^2,\lambda,\nu,\tau)
	=	
	{1\over\sigma}\,
	f_{\nu}(z)\, 
	\dfrac{ 
		F_{\nu+1}\Big( \left(\lambda z+\tau\right)
		\sqrt{\nu+1\over\nu+z^2}\,\Big) 
	}{
		F_{\nu}\Big({\tau\over\sqrt{1+\lambda^2}}\Big)
	},  
	\quad x\in\mathbb{R}; \  \mu,\lambda,\tau\in\mathbb{R}, \ \sigma,\nu>0,
	\end{align*}
	where $z=(x-\mu)/\sigma$, and $f_{\nu}$ and $F_\nu$ denote the PDF and CDF of the standard Student-$t$ distribution with $\nu>0$ degrees of freedom, respectively. 	Let $S_{{\rm ESN}_1}(x;\mu,\sigma^2,\lambda,\tau)$ be the SF corresponding to EST$_1$ PDF.
\end{definition}

\medskip
From \cite{Avg:10}, the  unconditional and conditional distributions of  $\boldsymbol{W}=\boldsymbol{X}\,\vert\, \boldsymbol{\lambda}^\top (\boldsymbol{X}-\boldsymbol{\mu})+\tau>Z$ are respectively given by
\begin{align}
&W_{i}
\sim
{\rm EST}_1\left(
\mu_{i}, \,
\sigma_{ii}, \,
\dfrac{\lambda_i \sigma_{ii}^{1/2}+\lambda_j \sigma_{jj}^{1/2}\rho_{ij}}{\sigma_{ii}^{1/2}\sqrt{1+\lambda_j^2\sigma_{jj}(1-\rho_{ij}^2)}}, \,
\nu, \,
\dfrac{\tau}{\sqrt{1+\lambda_j^2\sigma_{jj}(1-\rho_{ij}^2)}}
\right), 
\label{eq-1}
\\[0,2cm]
&W_{j}
\sim
{\rm EST}_1\left(
\mu_{j}, \,
\sigma_{jj}, \,
\dfrac{\lambda_j\sigma_{jj}^{1/2}+\lambda_i\sigma_{ii}^{1/2}\rho_{ij}}{\sigma_{jj}^{1/2}\sqrt{1+\lambda_{ii}\sigma_{ii}(1-\rho_{ij}^2)}}, \,
\nu, \,
\dfrac{\tau}{\sqrt{1+\lambda_i^2\sigma_{ii}(1-\rho_{ij}^2)}}
\right),
\label{eq-2}
\end{align}
and
\begin{align}
W_{j}\,\vert\,W_{i}=y
\sim 
{\rm EST}_1\left(
\boldsymbol{\mu}_{y}, \,
\boldsymbol{\sigma}^{\,2}_{y;\nu}, \,
\lambda_j\sigma_{jj}^{1/2}\sqrt{1-\rho_{ij}^2},\,
\nu+1,\,
\boldsymbol{\tau}_{y;\nu}
\right),
\label{eq-2-5}
\end{align}
where we are adopting the following notation:
\begin{align}\label{eq-3}
\begin{array}{lll}
&\displaystyle
\boldsymbol{\mu}_{y}
=
\mu_j+\sigma_{jj}^{1/2}\rho_{ij}\left(y-\mu_{i}\over \sigma_{ii}^{1/2}\right);
\\[0,6cm]
&\displaystyle
\boldsymbol{\sigma}^{\,2}_{y;\nu}
=
\dfrac{
	\nu+{(y-\mu_{1i})^2\over \sigma_{ii}}
}{\nu+1}\, \sigma_{jj} (1-\rho_{ij}^2);
\\[0,6cm]
\displaystyle
&\boldsymbol{\tau}_{y;\nu}
=
\left[(\lambda_i\sigma_{ii}^{1/2}+\lambda_j\sigma_{jj}^{1/2}\rho_{ij})\left(y-\mu_{i}\over \sigma_{ii}^{1/2}\right)+\tau\right]
\sqrt{\nu+1\over \nu+{(y-\mu_{i})^2\over \sigma_{ii}}}.
\end{array} 
\end{align}

Hence, by combining \eqref{main-formula} with \eqref{eq-2}, \eqref{eq-2-5} and \eqref{eq-3},  we obtain
	\begin{align}\label{main-formula-1}
f_{Y_i\,|\,Y_j > \kappa}(y)
	&=
	f_{{\rm EST}_1}\left(
	G_i(y); \,
	\mu_{i},  \,
	\sigma_{ii},  \,
\dfrac{\lambda_i \sigma_{ii}^{1/2}+\lambda_j\sigma_{jj}^{1/2}\rho_{ij}}{\sigma_{ii}^{1/2}\sqrt{1+\lambda_j\sigma_{jj}(1-\rho_{ij}^2)}}, \,
\nu, \,
\dfrac{\tau}{\sqrt{1+\lambda_j^2\sigma_{jj}(1-\rho_{ij}^2)}}
	\right) 
	G_{i}'(y)
	\nonumber
	\\[0,3cm]
	&\times 
	\dfrac{\displaystyle
		S_{{\rm EST}_1}\left(
		G_j(\kappa); \,
\boldsymbol{\mu}_{_{G_i(y)}}, \,
\boldsymbol{\sigma}^{\,2}_{_{G_i(y);\nu}}, \,
\lambda_j\sigma_{jj}^{1/2}\sqrt{1-\rho_{ij}^2},\,
\nu+1,\,
\boldsymbol{\tau}_{_{G_i(y);\nu}}
		\right) 
	}{ 
		S_{{\rm EST}_1}\left(
		G_j(\kappa); \,
		\mu_{j}, \,
		\sigma_{jj}, \,
\dfrac{\lambda_j\sigma_{jj}^{1/2}+\lambda_i\sigma_{ii}^{1/2}\rho_{ij}}{\sigma_{jj}^{1/2}\sqrt{1+\lambda_i^2\sigma_{ii}(1-\rho_{ij}^2)}}, \,
\nu, \,
\dfrac{\tau}{\sqrt{1+\lambda_i^2\sigma_{ii}(1-\rho_{ij}^2)}}
		\right)
	},
	\end{align}
for $y\in D$ and $\kappa\in D$.

On the other hand, from \eqref{marginal-pdf} and \eqref{eq-1}  the marginal PDF of $Y_i$ is obtained.

\subsubsection{Gaussian density generator}
Let $g^{(n)}(x)=\exp(-x/2)$, $x\in\mathbb{R}$ (see Table \ref{table:1}), 
be the Gaussian density generator of the \text{EGSE}$_n$ (multivariate extended $G$-skew-normal) distribution.

\begin{definition}
	A random variable $X$ follows a univariate extended skew-normal (ESN$_1$) distribution, denoted by $X\sim{ ESN}_1(\mu,\sigma^2,\lambda,\tau)$, if its PDF is given by 
	\cite[see][]{Vernic:05, Avg:10}
	\begin{align*}
	f_{{\rm ESN}_1}(x;\mu,\sigma^2,\lambda,\tau)
	=	
	{1\over\sigma}\,
	\phi(z)\, 
	{\Phi(\lambda z+\tau)\over \Phi\big({\tau\over\sqrt{1+\lambda^2}}\big)}, 
	\quad x\in\mathbb{R}; \ \mu,\lambda,\tau\in\mathbb{R}, \ \sigma>0,
	\end{align*}
	where $z=(x-\mu)/\sigma$, and $\phi$ and $\Phi$ denote the PDF and CDF of the standard normal  distribution, respectively.	Let $S_{{\rm ESN}_1}(x;\mu,\sigma^2,\lambda,\tau)$ denote the SF corresponding to ESN$_1$ PDF.
\end{definition}

\medskip 
Since 
\begin{align*}
\lim_{\nu\to\infty}\boldsymbol{\sigma}^{\,2}_{y;\nu}
= 
\sigma_{jj}(1-\rho_{ij}^2),
\quad 
\lim_{\nu\to\infty}\boldsymbol{\tau}_{y;\nu}
=
(\lambda_i\sigma_{ii}^{1/2}+\lambda_j\sigma_{jj}^{1/2}\rho_{ij})\left(y-\mu_{i}\over \sigma_{ii}^{1/2}\right)+\tau,
\end{align*}
and
$\lim_{\nu\to\infty}	f_{{\rm EST}_1}(x;\mu,\sigma^2,\lambda,\nu,\tau)
=
f_{{\rm ESN}_1}(x;\mu,\sigma^2,\lambda,\tau)$, by letting $\nu\to\infty$ in \eqref{main-formula-1}, we obtain
{\small
	\begin{align}\label{main-formula-2}
&f_{Y_i\,|\,Y_j > \kappa}(y)
	=
	f_{{\rm ESN}_1}\left(
	G_i(y); \,
	\mu_{i}, \,
	\sigma_{ii}, \,
\dfrac{\lambda_i \sigma_{ii}^{1/2}+\lambda_j\sigma_{jj}^{1/2}\rho_{ij}}{\sigma_{ii}^{1/2}\sqrt{1+\lambda_j^2\sigma_{jj}(1-\rho_{ij}^2)}}, \,
\nu, \,
\dfrac{\tau}{\sqrt{1+\lambda_j^2\sigma_{jj}(1-\rho_{ij}^2)}}
	\right) 		
	G_{i}'(y) 
	\nonumber
	\\[0,3cm]
	&
	\times 
	\dfrac{
		S_{{\rm ESN}_1}\left(
		G_j(\kappa); \,
		\mu_j+\sigma_{jj}^{1/2}\rho_{ij}\left(G_i(y)-\mu_{i}\over \sigma_{ii}^{1/2}\right), \, 
		\sigma_{jj}(1-\rho_{ij}^2), \,
		\lambda_j\sigma_{jj}^{1/2}\sqrt{1-\rho_{ij}^2}, \,
		(\lambda_i\sigma_{ii}^{1/2}+\lambda_j\sigma_{jj}^{1/2}\rho_{ij})\left(G_i(y)-\mu_{i}\over \sigma_{ii}^{1/2}\right)+\tau
		\right) 
	}{ 
		S_{{\rm ESN}_1}\left(
		G_j(\kappa); \,
		\mu_{j}, \,
		\sigma_{jj}, \,
\dfrac{\lambda_j\sigma_{jj}^{1/2}+\lambda_i\sigma_{ii}^{1/2}\rho_{ij}}{\sigma_{jj}^{1/2}\sqrt{1+\lambda_i^2\sigma_{ii}(1-\rho_{ij}^2)}}, \,
\nu, \,
\dfrac{\tau}{\sqrt{1+\lambda_i^2\sigma_{ii}(1-\rho_{ij}^2)}}
		\right)
	},
	\end{align}
}\noindent
for $y\in D$ and $\kappa\in D$.

On the other hand, from \eqref{marginal-pdf} and \eqref{eq-1} (with $\nu\to\infty$)  the marginal PDF of $Y_i$ is obtained.

%

%

\subsection{Expected value of a function of an EGSE$_n$ random vector
}


Let $\boldsymbol{Y}=(Y_1,\ldots,Y_n)^\top\sim\text{ EGSE}_n(\boldsymbol{\mu},\boldsymbol{\Sigma},\boldsymbol{\lambda},\tau,g^{(n)})$ and let
$\varphi: D^n\to\mathbb{R}$ be a real-valued measurable-analytic function. In this subsection, we provide simple closed formulas for the expected value of $\varphi(\boldsymbol{Y})$ and for the  mixed-moments, marginal moments and cross-moments of the EGSE$_n$ random vector $\boldsymbol{Y}$ for the special case $G_i(x)=\log(x)$, $x\in D=(0,\infty)$, $i=1,\ldots,n$.

Indeed, 
from stochastic representation in \eqref{rep-stoch} it follows that
\begin{align*}
\varphi(\boldsymbol{Y})
\stackrel{d}{=}
\varphi(G_1^{-1}(W_1),\ldots,G_n^{-1}(W_n)),
\end{align*}
where $\boldsymbol{W}\sim\text{ ESE}_n(\boldsymbol{\mu},\boldsymbol{\Sigma},\boldsymbol{\lambda},\tau,g^{(n)})$.
Let $\psi=\varphi\circ (G_1^{-1}\circ\pi_1,\ldots,G_n^{-1}\circ\pi_n)$ denote the composition function of $\varphi$ with $(G_1^{-1}\circ\pi_1,\ldots,G_n^{-1}\circ\pi_n)$, where $\pi_k$ denotes the $k$th projection function. The above representation is written as
\begin{align*}
\varphi(\boldsymbol{Y})
\stackrel{d}{=}
\psi(\boldsymbol{W}),
\end{align*}
which implies that
%
\begin{align}\label{rel-1}
\mathbb{E}[\varphi(\bm Y)]
=
\mathbb{E}[\psi(\boldsymbol{W})]
=
\int_{\mathbb{R}^n}
\psi(\bm w)
f_{\boldsymbol{W}}(\boldsymbol{w})
{\rm d} \boldsymbol{w}.
\end{align}
Consider ${\bm v}=(v_1,\ldots,v_n)^\top\in\mathbb{R}^n$ an $n$-dimensional vector.
Upon using the multivariate Taylor expansion of function $\bm w\longmapsto \psi(\bm w)$ around the point $\bm v$, that is (committing an abuse of notation),
\begin{align}
\psi(\bm w+\bm v)
&=
\left(
\sum_{k=0}^\infty
{1\over k!}\,
\sum_{i_1,\ldots,i_k=1}^{n} w_{i_1}\cdots w_{i_k} \, {\partial^k\over\partial v_{i_1}\cdots v_{i_k}} 
\right)
\psi({\bm v})
\nonumber
\\[0,2cm]
&=
\left(
\sum_{k=0}^\infty
{1\over k!}\,
(\bm w^\top \boldsymbol{\nabla_{\bm v}})^{k} 
\right)
\psi(\bm v),
\quad 
\text{with} \
\boldsymbol{\nabla_{\bm v}}
=
\left({\partial\over\partial v_1},\ldots,{\partial\over\partial v_n}\right)^\top,
\nonumber
\\[0,2cm]
&=
\exp(\bm w^\top \boldsymbol{\nabla_{\bm v}}) \psi(\bm v),
\label{expansion-taylor}
\end{align}
the expectation in \eqref{rel-1} becomes
\begin{align}
\mathbb{E}[\varphi(\bm Y)]
&=
	\int_{\mathbb{R}^n}
	\left[
	\psi(\bm w+\bm v)
	\big\vert_{\bm v=\bm 0}
	\right]
	f_{\boldsymbol{W}}(\boldsymbol{w})
	{\rm d} \boldsymbol{w} \nonumber
		\\[0,2cm]
	&=
	\int_{\mathbb{R}^n}
\left[
\exp(\bm w^\top \boldsymbol{\nabla_{\bm v}}) \psi(\bm v)
\big\vert_{\bm v=\bm 0}
\right]
f_{\boldsymbol{W}}(\boldsymbol{w})
{\rm d} \boldsymbol{w}	
	\nonumber
	\\[0,2cm]
	&=
	\left[
		\int_{\mathbb{R}^n}
		\exp(\bm w^\top \boldsymbol{\nabla_{\bm v}})
			f_{\boldsymbol{W}}(\boldsymbol{w})
			{\rm d} \boldsymbol{w}
			\right]
	\psi(\bm v)\Bigg\vert_{\bm v=\bm 0} 
	=
M_{\boldsymbol{W}}(\boldsymbol{\nabla_{\bm v}})
\psi(\bm v)
	\big\vert_{\bm v=\bm 0},
	\label{formula-exp-function}
\end{align}
where
\begin{align}\label{def-psi}
\psi(\bm v)=\varphi(G_1^{-1}(v_1),\ldots,G_n^{-1}(v_n))
\end{align}
and $M_{\boldsymbol{W}}(\bm s)$ is the moment generating function (MGF) of the multivariate random vector $\boldsymbol{W}$, whenever it exists.

\medskip
In the case that $\boldsymbol{Y}$ has a multivariate extended $G$-skew-normal distribution (see Table \ref{table:1}) case, $\boldsymbol{W}$ follows an multivariate extended skew-normal distribution (see Table \ref{table:4}) with parameter vector $(\boldsymbol{\mu},\boldsymbol{\Sigma},\boldsymbol{\lambda},\tau)^\top$. So, by using the definition of PDF $f_{\boldsymbol{W}}$ given in \eqref{def-pdf-W}, we have
\begin{align*}
	 M_{\boldsymbol{W}}(\bm s)
	 &=
	 \int_{\mathbb{R}^n}
	 \exp(\bm s^\top \boldsymbol{w})
	 f_{\boldsymbol{W}}(\boldsymbol{w})
	 {\rm d} \boldsymbol{w}
	 \\[0,2cm]
	 &=
	 	\int_{\mathbb{R}^n}
	 \exp(\bm s^\top \boldsymbol{w})
	 \phi_n(\boldsymbol{w};\,\boldsymbol{\mu},\boldsymbol{\Sigma}) \,
	 {
	 	\Phi\left(\boldsymbol{\lambda}^\top (\boldsymbol{w}-\boldsymbol{\mu})+\tau\right)
	 	\over 
	 	\Phi\Big(
	 	{\tau
	 		\over \sqrt{1+\boldsymbol{\lambda}^\top\boldsymbol{\Sigma}\boldsymbol{\lambda}}
	 	}
	 	\Big)
	} {\rm d} \boldsymbol{w}.
\end{align*}

A simple observation shows that
\begin{align*}
\exp(\bm s^\top \boldsymbol{w})
\phi_n(\boldsymbol{w};\,\boldsymbol{\mu},\boldsymbol{\Sigma})
=
\exp\left(\bm s^\top \boldsymbol{\mu}+{1\over 2}\,\boldsymbol{s}^\top \boldsymbol{\Sigma}\boldsymbol{s}\right)
\phi_n(\boldsymbol{w};\,\boldsymbol{\mu}^*,\boldsymbol{\Sigma}),
\quad
\boldsymbol{\mu}^* 
=
\boldsymbol{\mu}+\boldsymbol{\Sigma}\boldsymbol{s}.
\end{align*}

Then, upon using the above identity, the MGF of $\boldsymbol{W}$ is 
\begin{align*}
	 M_{\boldsymbol{W}}(\bm s)
	 =
	 \exp\left(\bm s^\top \boldsymbol{\mu}+{1\over 2}\,\boldsymbol{s}^\top \boldsymbol{\Sigma}\boldsymbol{s}\right)
	 \,
	 {
	 		\Phi\big(
	 {\tau^*
	 	\over \sqrt{1+\boldsymbol{\lambda}^\top\boldsymbol{\Sigma}\boldsymbol{\lambda}}
	 }
	 \big)
	 \over 
	 \Phi\big(
	 {\tau
	 	\over \sqrt{1+\boldsymbol{\lambda}^\top\boldsymbol{\Sigma}\boldsymbol{\lambda}}
	 }
	 \big)
	}
		 \int_{\mathbb{R}^n}
\phi_n(\boldsymbol{w};\,\boldsymbol{\mu}^*,\boldsymbol{\Sigma}) \,
	{
		\Phi\left(\boldsymbol{\lambda}^\top (\boldsymbol{w}-\boldsymbol{\mu}^*)+\tau^*\right)
		\over 
		\Phi\Big(
		{\tau^*
			\over \sqrt{1+\boldsymbol{\lambda}^\top\boldsymbol{\Sigma}\boldsymbol{\lambda}}
		}
		\Big)
	} {\rm d} \boldsymbol{w},
\end{align*}
with $\tau^*=\boldsymbol{\lambda}^\top\boldsymbol{\Sigma}\boldsymbol{s}+\tau$. Let
$\boldsymbol{W}^*$ be a random vector following a multivariate extended skew-normal distribution (see Table \ref{table:4}) with parameter vector $(\boldsymbol{\mu}^*,\boldsymbol{\Sigma},\boldsymbol{\lambda},\tau^*)$.
Using this notation, the MGF of $\boldsymbol{W}$ is expressed as
\begin{align*}
M_{\boldsymbol{W}}(\bm s)
&=
\exp\left(\bm s^\top \boldsymbol{\mu}+{1\over 2}\,\boldsymbol{s}^\top \boldsymbol{\Sigma}\boldsymbol{s}\right)
\,
{
	\Phi\big(
	{\tau^*
		\over \sqrt{1+\boldsymbol{\lambda}^\top\boldsymbol{\Sigma}\boldsymbol{\lambda}}
	}
	\big)
	\over 
	\Phi\big(
	{\tau
		\over \sqrt{1+\boldsymbol{\lambda}^\top\boldsymbol{\Sigma}\boldsymbol{\lambda}}
	}
	\big)
}
	 \int_{\mathbb{R}^n}
f_{\boldsymbol{W}^*}(\boldsymbol{w})
{\rm d} \boldsymbol{w}
\\[0,2cm]
&=
{
1
	\over 
	\Phi\big(
	{\tau
		\over \sqrt{1+\boldsymbol{\lambda}^\top\boldsymbol{\Sigma}\boldsymbol{\lambda}}
	}
	\big)
}
\,
\exp\left(\bm s^\top \boldsymbol{\mu}+{1\over 2}\,\boldsymbol{s}^\top \boldsymbol{\Sigma}\boldsymbol{s}\right)
\,
{
	\Phi\left(
	{\boldsymbol{\lambda}^\top\boldsymbol{\Sigma}\boldsymbol{s}+\tau
		\over \sqrt{1+\boldsymbol{\lambda}^\top\boldsymbol{\Sigma}\boldsymbol{\lambda}}
	}
	\right)
}.
\end{align*}

Replacing the above formula in \eqref{formula-exp-function}, we have
\begin{align*}
	\mathbb{E}[\varphi(\bm Y)]
	=
	\left[
	\exp(
	\bm \nabla_{\bm v}^\top \boldsymbol{\mu}
	) \psi(\bm v)	\big\vert_{\bm v=\bm 0}\,
	\right]
	\left[
		\exp\left(
	{1\over 2}\boldsymbol{\nabla_{\bm v}}^\top \boldsymbol{\Sigma}\boldsymbol{\nabla_{\bm v}}
	\right) \psi(\bm v)	\Bigg\vert_{\bm v=\bm 0}\,
	\right]
	\left[
	{
		\Phi\left(
		{\boldsymbol{\lambda}^\top\boldsymbol{\Sigma}\boldsymbol{\nabla_{\bm v}}+\tau
			\over \sqrt{1+\boldsymbol{\lambda}^\top\boldsymbol{\Sigma}\boldsymbol{\lambda}}
		}
		\right)
			\over 
	\Phi\big(
	{\tau
		\over \sqrt{1+\boldsymbol{\lambda}^\top\boldsymbol{\Sigma}\boldsymbol{\lambda}}
	}
	\big)
}\,
\psi(\bm v)
	\Bigg\vert_{\bm v=\bm 0}
	\right].
\end{align*}
By using the multivariate Taylor expansion \eqref{expansion-taylor}, $	\exp(
\bm \nabla_{\bm v}^\top \boldsymbol{\mu}
) \psi(\bm v)=\psi(\bm \mu+\bm v)$. Then, we obtain the following closed formula for the expected value of a function of  $\bm Y$ having a multivariate extended $G$-skew-normal distribution (see Table \ref{table:1}):
\begin{align}
	\mathbb{E}[\varphi(\bm Y)]
=
\psi(\bm \mu)	
\left[
\exp\left(
{1\over 2}\boldsymbol{\nabla_{\bm v}}^\top \boldsymbol{\Sigma}\boldsymbol{\nabla_{\bm v}}
\right) \psi(\bm v)	\Bigg\vert_{\bm v=\bm 0}\,
\right]
\left[
{
	\Phi\left(
	{\boldsymbol{\lambda}^\top\boldsymbol{\Sigma}\boldsymbol{\nabla_{\bm v}}+\tau
		\over \sqrt{1+\boldsymbol{\lambda}^\top\boldsymbol{\Sigma}\boldsymbol{\lambda}}
	}
	\right)
	\over 
	\Phi\big(
	{\tau
		\over \sqrt{1+\boldsymbol{\lambda}^\top\boldsymbol{\Sigma}\boldsymbol{\lambda}}
	}
	\big)
}\,
\psi(\bm v)
\Bigg\vert_{\bm v=\bm 0}
\right],
\label{formula-exp-function-1}
\end{align}
with $\psi$ being as in \eqref{def-psi}.

\begin{remark}
	\begin{itemize}
	\item[(i)]
When the extension parameter is absent, that is, $\tau=0$, we have
\begin{align*}
	\mathbb{E}[\varphi(\bm Y)]
=
2
\psi(\bm \mu)	
\left[
\exp\left(
{1\over 2}\boldsymbol{\nabla_{\bm v}}^\top \boldsymbol{\Sigma}\boldsymbol{\nabla_{\bm v}}
\right) \psi(\bm v)	\Bigg\vert_{\bm v=\bm 0}\,
\right]
\left[
{
	\Phi\left(
	{\boldsymbol{\lambda}^\top\boldsymbol{\Sigma}\boldsymbol{\nabla_{\bm v}}
		\over \sqrt{1+\boldsymbol{\lambda}^\top\boldsymbol{\Sigma}\boldsymbol{\lambda}}
	}
	\right)
}\,
\psi(\bm v)
\Bigg\vert_{\bm v=\bm 0}\,
\right].
\end{align*}
	\item[(ii)]
When the skewness parameter is absent, that is, $\boldsymbol{\lambda}=\boldsymbol{0}$, we have
\begin{align*}
\mathbb{E}[\varphi(\bm Y)]
=
\psi(\bm \mu)	
\left[
\exp\left(
{1\over 2}\boldsymbol{\nabla_{\bm v}}^\top \boldsymbol{\Sigma}\boldsymbol{\nabla_{\bm v}}
\right) \psi(\bm v)	\Bigg\vert_{\bm v=\bm 0}\,
\right].
\end{align*}
\end{itemize}
\end{remark}

\begin{remark}
	\begin{itemize}
	\item[(i)]
The exponential operator $\exp\left(\boldsymbol{\nabla_{\bm v}}^\top \boldsymbol{\Sigma}\boldsymbol{\nabla_{\bm v}}/2\right)$ that appears in \eqref{formula-exp-function-1} can be written as
\begin{align}\label{operator-1}
\exp\left({1\over 2}\,\boldsymbol{\nabla_{\bm v}}^\top \boldsymbol{\Sigma}\boldsymbol{\nabla_{\bm v}}\right)
&=
\sum_{k=0}^{\infty}
{1\over k!}\,
\left(
{1\over 2}\,\boldsymbol{\nabla_{\bm v}}^\top \boldsymbol{\Sigma}\boldsymbol{\nabla_{\bm v}}\right)^{k}
\nonumber
\\[0,2cm]
&=
\sum_{k=0}^{\infty}
{1\over k!}\,
{1\over 2^k}
\sum_{j_1,l_1,\ldots,j_k,l_k=1}^n
\sigma_{j_1l_1}\cdots\sigma_{j_kl_k}\,{\partial^{2k}\over\partial v_{j_1} \partial v_{l_1}\cdots \partial v_{j_k} \partial v_{l_k}}.
\end{align}

	\item[(ii)]
By using the series representation of the Gaussian CDF:
\begin{align*}
\Phi(x)
=
{1\over 2}
+
{1\over\sqrt{\pi}}
\sum_{k=0}^\infty {(-1)^{3k} 2^{-{1\over 2}-k}\over (1+2k)k!}\, x^{2k},
\end{align*}
the operator $	\Phi((
\boldsymbol{\lambda}^\top\boldsymbol{\Sigma}\boldsymbol{\nabla_{\bm v}}+\tau)
	/ \sqrt{1+\boldsymbol{\lambda}^\top\boldsymbol{\Sigma}\boldsymbol{\lambda}}
\,)$ that appears in \eqref{formula-exp-function-1} can be written as
\begin{align}\label{operator-2}
{
	\Phi\left(
	{\boldsymbol{\lambda}^\top\boldsymbol{\Sigma}\boldsymbol{\nabla_{\bm v}}+\tau
		\over \sqrt{1+\boldsymbol{\lambda}^\top\boldsymbol{\Sigma}\boldsymbol{\lambda}}
	}
	\right)
}
&=
{1\over 2}
+
{1\over\sqrt{\pi}}
\sum_{k=0}^\infty {(-1)^{3k} 2^{-{1\over 2}-k}\over (1+2k)k!} \left(
{\boldsymbol{\lambda}^\top\boldsymbol{\Sigma}\boldsymbol{\nabla_{\bm v}}+\tau
	\over \sqrt{1+\boldsymbol{\lambda}^\top\boldsymbol{\Sigma}\boldsymbol{\lambda}}
}
\right)^{2k}
\nonumber
\\[0,2cm]
&=
{1\over 2}
+
{1\over\sqrt{\pi}}
\sum_{k=0}^\infty {(-1)^{3k} 2^{-{1\over 2}-k}\over (1+2k)k!} 
\sum_{r=0}^{2k}
\binom{2k}{r}
\left(
{
	\tau
	\over \sqrt{1+\boldsymbol{\lambda}^\top\boldsymbol{\Sigma}\boldsymbol{\lambda}}
}
\right)^{2k-r}
\nonumber
\\[0,2cm]
&\times
{\displaystyle 
	\sum_{j_1,l_1,\ldots,j_r,l_r=1}^{n} \sigma_{l_1j_1}\cdots\sigma_{l_rj_r}\lambda_{l_1}\cdots\lambda_{l_r}
	{\partial^r\over\partial v_{j_1}\cdots\partial v_{j_r}}
	\over (\sqrt{1+\boldsymbol{\lambda}^\top\boldsymbol{\Sigma}\boldsymbol{\lambda}}\,)^r
},
\end{align}
where in the last equality a binomial expansion was used.
\end{itemize}
\end{remark}

\begin{remark}
Since $	\mathbb{E}[\varphi(\bm Y)]$ in \eqref{formula-exp-function-1} depends on the operator formulas in \eqref{operator-1} and \eqref{operator-2}, these can be used to facilitate its calculation.
\end{remark}

\subsubsection{Mixed-moments}
Let $\varphi(\bm y)=\prod_{i=1}^{n} \pi^m_i(\bm y)=\prod_{i=1}^{n} y_i^{m_i}$, where $\pi_i$ is the  $i$th projection function. From \eqref{formula-exp-function} we have  the next formula for the mixed-moments of $\bm Y$:
\begin{align*}
\mathbb{E}\left(\prod_{i=1}^{n} Y_i^{m_i}\right)
=
M_{\boldsymbol{W}}(\boldsymbol{\nabla_{\bm v}})
\prod_{i=1}^{n}
[G_i^{-1}(v_i)]^{m_i}
\Big\vert_{\bm v=\bm 0}.
\end{align*}

\medskip
In the case that $\boldsymbol{Y}$ has a multivariate extended $G$-skew-normal distribution (see Table \ref{table:1}), 
from \eqref{formula-exp-function-1} we have 
\begin{align}\label{mixed-moments}
\mathbb{E}\left(\prod_{i=1}^{n} Y_i^{m_i}\right)
&=
\prod_{i=1}^{n}
[G_i^{-1}(\mu_i)]^{m_i}
\left[
\exp\left(
{1\over 2}\boldsymbol{\nabla_{\bm v}}^\top \boldsymbol{\Sigma}\boldsymbol{\nabla_{\bm v}}
\right) \prod_{i=1}^{n} [G_i^{-1}(v_i)]^{m_i}	\Bigg\vert_{\bm v=\bm 0}\,
\right]
\nonumber
\\[0,2cm]
&\times
\left[
{
	\Phi\left(
	{\boldsymbol{\lambda}^\top\boldsymbol{\Sigma}\boldsymbol{\nabla_{\bm v}}+\tau
		\over \sqrt{1+\boldsymbol{\lambda}^\top\boldsymbol{\Sigma}\boldsymbol{\lambda}}
	}
	\right)
	\over 
	\Phi\big(
	{\tau
		\over \sqrt{1+\boldsymbol{\lambda}^\top\boldsymbol{\Sigma}\boldsymbol{\lambda}}
	}
	\big)
}
\prod_{i=1}^{n}
[G_i^{-1}(v_i)]^{m_i}
\Bigg\vert_{\bm v=\bm 0}
\right].
\end{align}

\medskip 
It is clear that the above formula is extremely complicated for functions $G_i$s in general such as those in Table \ref{table:2-1}. For illustration purposes, let us consider $G_i(x)=\log(x)$, $x\in D=(0,\infty)$, $i=1,\ldots,n$. So, by using formula in \eqref{operator-1}, we have
\begin{align*}
	\exp\left({1\over 2}\,\boldsymbol{\nabla_{\bm v}}^\top \boldsymbol{\Sigma}\boldsymbol{\nabla_{\bm v}}\right)
	\prod_{i=1}^{n}
	[G_i^{-1}(v_i)]^{m_i}
=
	\exp\left({1\over 2}\,{\bm m}^\top \boldsymbol{\Sigma} {\bm m}+{\bm m}^\top{\bm v}\right).
\end{align*}
On the other hand, by using formula in \eqref{operator-2}, we obtain
\begin{align*}
\Phi\left(
	{\boldsymbol{\lambda}^\top\boldsymbol{\Sigma}\boldsymbol{\nabla_{\bm v}}+\tau
		\over \sqrt{1+\boldsymbol{\lambda}^\top\boldsymbol{\Sigma}\boldsymbol{\lambda}}
	}
	\right)
	\prod_{i=1}^{n}
[G_i^{-1}(v_i)]^{m_i}
=
	\Phi\left(
{\boldsymbol{\lambda}^\top\boldsymbol{\Sigma} {\bm m}+\tau
	\over \sqrt{1+\boldsymbol{\lambda}^\top\boldsymbol{\Sigma}\boldsymbol{\lambda}}
}
\right)
\exp({\bm m}^\top{\bm v}).
\end{align*}
Replacing the last two expressions in \eqref{mixed-moments}, we obtain 
\begin{align*}
	\mathbb{E}\left(\prod_{i=1}^{n} Y_i^{m_i}\right)
	=
	\exp\left({\bm m}^\top{\bm \mu}+{1\over 2}\,{\bm m}^\top \boldsymbol{\Sigma} {\bm m}\right)
	\dfrac{
		\Phi\left(
	{\boldsymbol{\lambda}^\top\boldsymbol{\Sigma} {\bm m}+\tau
		\over \sqrt{1+\boldsymbol{\lambda}^\top\boldsymbol{\Sigma}\boldsymbol{\lambda}}
	}
	\right)
}{		\Phi\left(
{\tau
	\over \sqrt{1+\boldsymbol{\lambda}^\top\boldsymbol{\Sigma}\boldsymbol{\lambda}}
}
\right)}.
\end{align*}
The above formula has appeared in \cite{Genton2010} for the special case $\tau=0$.  In particular,
\begin{align*}
	\mathbb{E}\left(Y_i^{m}\right)
=
\exp\left(m\mu_i+{1\over 2}\ m^2\sigma_{ii}\right)
\dfrac{
	\Phi\left(
	{m\sum_{k=1}^{n}\lambda_k\sigma_{ki}+\tau
		\over \sqrt{1+\boldsymbol{\lambda}^\top\boldsymbol{\Sigma}\boldsymbol{\lambda}}
	}
	\right)
}{		\Phi\left(
	{\tau
		\over \sqrt{1+\boldsymbol{\lambda}^\top\boldsymbol{\Sigma}\boldsymbol{\lambda}}
	}
	\right)},
\quad i=1,\ldots,n.
\end{align*}

\begin{remark}
	In the case that $\boldsymbol{Y}$ has a multivariate extended $G$-skew-Student-$t$ distribution (see Table \ref{table:1}),  we cannot guarantee in general the existence of mixed-moments (in particular, the existence of moments), because in this case, when considering $G_i(x)=\log(x)$, $x\in D=(0,\infty)$, $i=1,\ldots,n$ and $\tau=0$, these moments do not exist (see Proposition 7 of reference \cite{Genton2010}).
\end{remark}

\subsubsection{Marginal moments}\label{Marginal moments}

Let $\varphi$ be the $i$th projection function raised to the  $m$th power, that is, $\varphi(\bm y)=\pi^m_i(\bm y)=y_i^m$, $i=1,\ldots,n$. From \eqref{formula-exp-function} we have  the next formula for the marginal moments of $\bm Y$:
\begin{align*}
	\mathbb{E}(Y_i^m)
	=
M_{\boldsymbol{W}}(\boldsymbol{\nabla_{\bm v}})
[G_i^{-1}(v_i)]^m
\big\vert_{v_i=0}.
\end{align*}

\medskip
In the case that $\boldsymbol{Y}$ has a multivariate extended $G$-skew-normal distribution (see Table \ref{table:1}) case, 
from \eqref{formula-exp-function-1} we have (for $i=1,\ldots,n$) 
\begin{align}\label{marginal-moments}
\mathbb{E}(Y_i^m)
=
[G_i^{-1}(\mu_i)]^m
\left[
\exp\left(
{1\over 2}\boldsymbol{\nabla_{\bm v}}^\top \boldsymbol{\Sigma}\boldsymbol{\nabla_{\bm v}}
\right) [G_i^{-1}(v_i)]^m	\Bigg\vert_{v_i=0}\,
\right]
\left[
{
	\Phi\left(
	{\boldsymbol{\lambda}^\top\boldsymbol{\Sigma}\boldsymbol{\nabla_{\bm v}}+\tau
		\over \sqrt{1+\boldsymbol{\lambda}^\top\boldsymbol{\Sigma}\boldsymbol{\lambda}}
	}
	\right)
	\over 
	\Phi\big(
	{\tau
		\over \sqrt{1+\boldsymbol{\lambda}^\top\boldsymbol{\Sigma}\boldsymbol{\lambda}}
	}
	\big)
}\,
[G_i^{-1}(v_i)]^m
\Bigg\vert_{v_i=0}
\right].
\end{align}
By using formula in \eqref{operator-1}, we have
\begin{align}\label{marginal-moments1}
\exp\left({1\over 2}\,\boldsymbol{\nabla_{\bm v}}^\top \boldsymbol{\Sigma}\boldsymbol{\nabla_{\bm v}}\right)
[G_i^{-1}(v_i)]^m
=
\exp\left({\sigma_{ii}^{2}\over 2}\, {\partial^{2} \over\partial v_{i}^{2}}\right) [G_i^{-1}(v_i)]^m.
\end{align}
On the other hand, by using formula in \eqref{operator-2}, we obtain
\begin{align} \label{marginal-moments2}
&{
	\Phi\left(
	{\boldsymbol{\lambda}^\top\boldsymbol{\Sigma}\boldsymbol{\nabla_{\bm v}}+\tau
		\over \sqrt{1+\boldsymbol{\lambda}^\top\boldsymbol{\Sigma}\boldsymbol{\lambda}}
	}
	\right)
}
[G_i^{-1}(v_i)]^m
=
\Phi
\left(
{
(\sum_{l=1}^{n} \sigma_{li}\lambda_l)
{\partial \over \partial v_i}
	+
	\tau
	\over \sqrt{1+\boldsymbol{\lambda}^\top\boldsymbol{\Sigma}\boldsymbol{\lambda}}
}
\right)
[G_i^{-1}(v_i)]^m.
\end{align}
Replacing the expressions  \eqref{marginal-moments1} and  \eqref{marginal-moments2} in \eqref{marginal-moments}, we obtain the following simple closed formula for the marginal moments of the multivariate extended skew-normal random vector $\boldsymbol{Y}$:
\begin{align}
\mathbb{E}(Y_i^m)
=
[G_i^{-1}(\mu_i)]^m
\left[
\exp\left({\sigma_{ii}^{2}\over 2}\, {\partial^{2} \over\partial v_{i}^{2}}\right) [G_i^{-1}(v_i)]^m	\Bigg\vert_{v_i=0}\,
\right]
\left[
{
	\Phi\left(
	{
(\sum_{l=1}^{n} \sigma_{li}\lambda_l)
{\partial \over \partial v_i}
+
\tau
\over 
\sqrt{1+\boldsymbol{\lambda}^\top\boldsymbol{\Sigma}\boldsymbol{\lambda}}
}
	\right)
	\over 
	\Phi\big(
	{\tau
		\over \sqrt{1+\boldsymbol{\lambda}^\top\boldsymbol{\Sigma}\boldsymbol{\lambda}}
	}
	\big)
}\,
[G_i^{-1}(v_i)]^m
\Bigg\vert_{v_i=0}
\right].
\end{align}

\subsubsection{Cross-moments}
By considering  $\varphi(\bm y)=\pi_i(\bm y)\pi_j(\bm y)=y_iy_j$, $i\neq j=1,\ldots,n$, where $\pi_k$ denotes the $k$th projection function, from \eqref{formula-exp-function} we have the following formula for the cross-moments of $\bm Y$: 
\begin{align*}
\mathbb{E}(Y_iY_j)
=
M_{\boldsymbol{W}}(\boldsymbol{\nabla_{\bm v}})
G_i^{-1}(v_i)
G_j^{-1}(v_j)
\big\vert_{v_i=v_j=0}.
\end{align*}

\medskip
In the case that $\boldsymbol{Y}$ has a multivariate extended $G$-skew-normal distribution (see Table \ref{table:1}) case, 
from \eqref{formula-exp-function-1} we have 
\begin{align}\label{cross-moments}
\mathbb{E}(Y_i Y_j)
&=
G_i^{-1}(\mu_i)G_j^{-1}(\mu_j)\!
\left[
\exp\left({1\over 2}\boldsymbol{\nabla_{\bm v}}^\top \boldsymbol{\Sigma}\boldsymbol{\nabla_{\bm v}}\right)
G_i^{-1}(v_i)G_j^{-1}(v_j)\,
\Bigg\vert_{v_i=v_j= 0}
\right]
\nonumber
\\[0,2cm]
&\times
\left[
{
	\Phi\left(
	{\boldsymbol{\lambda}^\top\boldsymbol{\Sigma}\boldsymbol{\nabla_{\bm v}}+\tau
		\over \sqrt{1+\boldsymbol{\lambda}^\top\boldsymbol{\Sigma}\boldsymbol{\lambda}}
	}
	\right)
	\over 
	\Phi\big(
	{\tau
		\over \sqrt{1+\boldsymbol{\lambda}^\top\boldsymbol{\Sigma}\boldsymbol{\lambda}}
	}
	\big)
}\,
G_i^{-1}(v_i)G_j^{-1}(v_j)\,
\Bigg\vert_{v_i=v_j= 0}
\right].
\end{align}
By using formula in \eqref{operator-1}, we have
\begin{align}\label{marginal-moments3}
\exp\left({1\over 2}\,\boldsymbol{\nabla_{\bm v}}^\top \boldsymbol{\Sigma}\boldsymbol{\nabla_{\bm v}}\right)
G_i^{-1}(v_i)G_j^{-1}(v_j)
=
\exp\left(
{1\over 2}
\sum_{r,s\in\{i,j\}}
\sigma_{rs}\,
{\partial^2\over\partial v_r \partial v_s}
\right)
G_i^{-1}(v_i)G_j^{-1}(v_j).
\end{align}
Furthermore, by using formula in \eqref{operator-2}, we obtain
\begin{align} \label{marginal-moments4}
&{
	\Phi\left(
	{\boldsymbol{\lambda}^\top\boldsymbol{\Sigma}\boldsymbol{\nabla_{\bm v}}+\tau
		\over \sqrt{1+\boldsymbol{\lambda}^\top\boldsymbol{\Sigma}\boldsymbol{\lambda}}
	}
	\right)
}
G_i^{-1}(v_i)G_j^{-1}(v_j)
=
\Phi\left({
	\left(\sum_{l=1}^{n} 
	\sigma_{li}\lambda_l\right)
	{\partial\over \partial v_i}
	+
	\left(\sum_{l=1}^{n} 
	\sigma_{lj}\lambda_l\right)
	{\partial\over \partial v_j}
	+
	\tau
	\over \sqrt{1+\boldsymbol{\lambda}^\top\boldsymbol{\Sigma}\boldsymbol{\lambda}}
}\right).
\end{align}
Replacing the expressions  \eqref{marginal-moments3} and  \eqref{marginal-moments4} in \eqref{cross-moments}, we obtain the following closed formula for the cross-moments  of the multivariate extended skew-normal random vector $\boldsymbol{Y}$:
\begin{align*}
\mathbb{E}(Y_i Y_j)
&=
G_i^{-1}(\mu_i)G_j^{-1}(\mu_j)
\left[
\exp\left(
{1\over 2}
\sum_{r,s\in\{i,j\}}
\sigma_{rs}\,
{\partial^2\over\partial v_r \partial v_s}
\right)
G_i^{-1}(v_i)G_j^{-1}(v_j)\,
\Bigg\vert_{v_i=v_j= 0}
\right]
\\[0,2cm]
&\times
\left[
{
\Phi\left({
	\left(\sum_{l=1}^{n} 
	\sigma_{li}\lambda_l\right)
	{\partial\over \partial v_i}
	+
	\left(\sum_{l=1}^{n} 
	\sigma_{lj}\lambda_l\right)
	{\partial\over \partial v_j}
	+
	\tau
	\over \sqrt{1+\boldsymbol{\lambda}^\top\boldsymbol{\Sigma}\boldsymbol{\lambda}}
}\right)
	\over 
	\Phi\big(
	{\tau
		\over \sqrt{1+\boldsymbol{\lambda}^\top\boldsymbol{\Sigma}\boldsymbol{\lambda}}
	}
	\big)
}\,
G_i^{-1}(v_i)G_j^{-1}(v_j)\,
\Bigg\vert_{v_i=v_j= 0}
\right],
\quad i\neq j=1,\ldots,n.
\end{align*}


\subsection{Existence of marginal moments when $D=(0,\infty)$} 
The objective of this subsection is to provide sufficient conditions to ensure the existence of the real moments of the random variable  
$Y_i
=
T_i\,\vert\, \boldsymbol{\lambda}^\top (\boldsymbol{X}-\boldsymbol{\mu})+\tau>Z$, with $T_i=G_i^{-1}(X_i)$ and $G_i:D=(0,\infty)\to\mathbb{R}$, $i=1,\ldots,n$.
To do this, we will consider the notation 
$
W_i
=
X_i\,\vert\, \boldsymbol{\lambda}^\top (\boldsymbol{X}-\boldsymbol{\mu})+\tau>Z
$, $i=1,\ldots,n$, used in Subsection \ref{Stochastic representation}.

Indeed, by using the well-known identity
\begin{align}\label{identity-known}
	\mathbb{E}(Y^p)=p\int_0^\infty y^{p-1}\mathbb{P}(Y>y){\rm d}y, 
	\quad Y>0, \ p>0,
\end{align}
and by employing the relation given in \eqref{identity-CDFs-univariate}:
\begin{align*}
	Y_i\stackrel{d}{=}G_i^{-1}(W_i), \quad i=1,\ldots,n,
\end{align*}
it follows that
\begin{align*}
	\mathbb{E}(Y_i^p)
	&=
	p\int_0^\infty y^{p-1}\mathbb{P}(W_i>G_i(y)){\rm d}y
	\\[0,2cm]
	&=
	p\int_0^a y^{p-1}\mathbb{P}(W_i>G_i(y)){\rm d}y
	+
	p\int_a^\infty y^{p-1}\mathbb{P}(W_i>G_i(y)){\rm d}y
	\\[0,2cm]
	&\leqslant
	a^p+p\int_a^\infty y^{p-1}\mathbb{P}(W_i>G_i(y)){\rm d}y,
\end{align*}
for some $a\in(0,\infty)$. Therefore, a sufficient condition for the existence of positive order moments of $Y_i$ is that
\begin{align}\label{cond-1}
	I=
	\int_a^\infty y^{p-1}\mathbb{P}(W_i>G_i(y)){\rm d}y<\infty, \quad i=1,\ldots,n.
\end{align}

\smallskip
In what remains of this subsection we will analyze condition in \eqref{cond-1} in the special case that  (see Table \ref{table:2-1})
\begin{align}\label{def-G}
	G_i(x)={2H_i(x)-1\over H_i(x)[1-H_i(x)]}, \quad x>0, \ i=1,\ldots,n,
\end{align}
with $H_i$ being the CDF of a continuous random variable with positive support. Indeed, 
as $\{W_i>G_i(y)\}\subset\{\vert W_i\vert>G_i(y)\}$, the integral in \eqref{cond-1} is
\begin{align*}
	I\leqslant
	\int_a^\infty y^{p-1}\mathbb{P}(\vert W_i\vert>G_i(y)){\rm d}y.
\end{align*}
By Markov's inequality, the above integral is at most
\begin{align*}
	\mathbb{E}(\vert W_i\vert^p)
	\int_a^\infty {y^{p-1}\over G^{p}_i(y)}{\rm d}y
	=
	\mathbb{E}(\vert W_i\vert^p)
	\int_a^\infty {y^{p-1}\over G^{p-1}_i(y)}
	{H_i(y)[1-H_i(y)] \over  [2H_i(y)-1]}{\rm d}y.
\end{align*}
As $G_i$ and $H_i$ are increasing, for $p>1$, the above expression is
\begin{align*}
	&\leqslant
	{\mathbb{E}(\vert W_i\vert^p)\over G^{p-1}_i(a) [2H_i(a)-1]}
	\int_a^\infty y^{p-1}
	{[1-H_i(y)]}{\rm d}y
	\\[0,2cm]
	&\leqslant
	{\mathbb{E}(\vert W_i\vert^p)\over G^{p-1}_i(a) [2H_i(a)-1]}
	\int_0^\infty y^{p-1}
	{[1-H_i(y)]}{\rm d}y,
\end{align*}
provided $H_i(a)\neq 1/2$ and $G_i(a)\in(0,\infty)$. If $S_i>0$ is a continuous random variable such that $S_i\stackrel{d}{=}H_i$, by \eqref{identity-known}, the above integral is
\begin{align*}
	=
	{\mathbb{E}(\vert W_i\vert^p)\mathbb{E}(S_i^p)\over pG^{p-1}_i(a) [2H_i(a)-1]}.
\end{align*}
Therefore, for the choice of $G_i$ as in \eqref{def-G}, we have verified that
\begin{align*}
	I\leqslant {\mathbb{E}(\vert W_i\vert^p)\mathbb{E}(S_i^p)\over pG^{p-1}_i(a) [2H_i(a)-1]}.
\end{align*}

Hence, if $G_i$ as in \eqref{def-G}, $a>0$ is such that $H_i(a)\neq 1/2$ and $G_i(a)\in(0,\infty)$, $\mathbb{E}(\vert W_i\vert^p)<\infty$ and $\mathbb{E}(S_i^p)<\infty$ for some $p>1$, then $\mathbb{E}(Y_i^p)$, $i=1,\ldots,n$, exists.


\begin{remark}
	The arguments given in this subsection can easily be extended to establish sufficient conditions for the existence of marginal moments when $D=(-\infty,\infty)$. 
\end{remark}

\subsection{Kullback-Leibler Divergence}

If $f_{\boldsymbol{Y}_1}$ and $f_{\boldsymbol{Y}_2}$ are the PDFs of $\boldsymbol{Y}_1=(Y_{11},\ldots,Y_{1n})^\top\sim \text{ EGSE}_n(\boldsymbol{\mu}_1,\boldsymbol{\Sigma}_1,\boldsymbol{\lambda}_1,\tau_1,g^{(n)})$  and $\boldsymbol{Y}_2=(Y_{21},\ldots,Y_{2n})^\top\sim \text{ EGSE}_n(\boldsymbol{\mu}_2,\boldsymbol{\Sigma}_2,\boldsymbol{\lambda}_2,\tau_2,g^{(n)})$, respectively, their Kullback-Leibler 
divergence measure is defined by
\begin{align*}
D_{\rm KL}(f_{\boldsymbol{Y}_1}\Vert f_{\boldsymbol{Y}_2})
=
\int_{D^n} 
f_{\boldsymbol{Y}_1}(\boldsymbol{y};\boldsymbol{\mu}_1,\boldsymbol{\Sigma}_1,\boldsymbol{\lambda}_1,\tau_1)
\log\left({f_{\boldsymbol{Y}_1}(\boldsymbol{y};\boldsymbol{\mu}_1,\boldsymbol{\Sigma}_1,\boldsymbol{\lambda}_1,\tau_1)
\over f_{\boldsymbol{Y}_2}(\boldsymbol{y};\boldsymbol{\mu}_2,\boldsymbol{\Sigma}_2,\boldsymbol{\lambda}_2,\tau_2)}\right)
{\rm d}{\bm y}.
\end{align*}
Since this divergence measure is invariant under invertible transforms, from stochastic representation in \eqref{rep-stoch}, we have
\begin{align*}
D_{\rm KL}(f_{\boldsymbol{Y}_1}\Vert f_{\boldsymbol{Y}_2})
=
D_{\rm KL}(f_{G_1^{-1}(W_{11}),\ldots,G_n^{-1}(W_{1n})}\Vert f_{G_1^{-1}(W_{21}),\ldots,G_n^{-1}(W_{2n})})
=
D_{\rm KL}(f_{\boldsymbol{W}_1}\Vert f_{\boldsymbol{W}_2}),
\end{align*}
where $f_{\boldsymbol{W}_1}$ and $f_{\boldsymbol{W}_2}$ are the PDFs of
$\boldsymbol{W}_1=(W_{11},\ldots,W_{1n})^\top\sim\text{ ESE}_n(\boldsymbol{\mu}_1,\boldsymbol{\Sigma}_1,\boldsymbol{\lambda}_1,\tau_1,g^{(n)})$
and
$\boldsymbol{W}_2=(W_{21},\ldots,W_{2n})^\top\sim\text{ ESE}_n(\boldsymbol{\mu}_2,\boldsymbol{\Sigma}_2,\boldsymbol{\lambda}_2,\tau_2,g^{(n)})$, respectively.
The Kullback-Leibler divergence measure $D_{\rm KL}(f_{\boldsymbol{W}_1}\Vert f_{\boldsymbol{W}_2})$ for $\boldsymbol{W}_1$ and $\boldsymbol{W}_2$ following multivariate extended skew-normal distributions, with $\tau=0$, was studied in detail in reference \cite{Contreras12}.

Note that, for $\bm \lambda=0$ and $\tau=0$, the Kullback-Leibler 
divergence for $f_{\boldsymbol{Y}_1}$ and $f_{\boldsymbol{Y}_2}$ reduces to 
\begin{align*}
	D_{\rm KL}(f_{\boldsymbol{Y}_1}\Vert f_{\boldsymbol{Y}_2})
	=
	D_{\rm KL}(f_{\boldsymbol{X}_1}\Vert f_{\boldsymbol{X}_2}),
\end{align*}
where
$\boldsymbol{X}_1=(X_{11},\ldots,X_{1n})^\top\sim \text{ ELL}_n(\boldsymbol{\mu}_1,\boldsymbol{\Sigma}_1,g^{(n)})$
and
$\boldsymbol{X}_2=(X_{21},\ldots,X_{2n})^\top\sim \text{ ELL}_n(\boldsymbol{\mu}_2,\boldsymbol{\Sigma}_2,g^{(n)})$.

%

\subsection{Maximum likelihood estimation}\label{sec:mle}

Let $\{\boldsymbol{Y}_k=(Y_{1k},Y_{2k},\ldots,Y_{nk})^\top: k=1,\ldots,m\}$ be a multivariate random sample of size $m$ from $\boldsymbol{Y}\sim\text{ EGSE}_n(\boldsymbol{\mu},\boldsymbol{\Sigma},\boldsymbol{\lambda},\tau,g^{(n)})$ with joint PDF as given in \eqref{def-pdf-T*}, and let $\boldsymbol{y}_k=(y_{1k},y_{2k},\ldots,y_{nk})^{\top}$ be a realization of $\boldsymbol{Y}_k$. To obtain the maximum likelihood estimates (MLEs) of the model parameters with parameter vector $\boldsymbol{\theta} = (\boldsymbol{\mu},\boldsymbol{\Sigma},\boldsymbol{\lambda},\tau)^{\top}$, we maximize the following log-likelihood function 
\begin{align*}
\ell(\boldsymbol{\theta})
&=
\sum_{k=1}^{m}
\log(f_{\boldsymbol{X}}(\boldsymbol{y}_{G,k}))
	+
	\sum_{k=1}^{m}
	\log(
		F_{{\rm ELL}_1}(\boldsymbol{\lambda}^\top (\boldsymbol{y}_{G,k}-\boldsymbol{\mu})+\tau;\, 0,1,g_{q(\boldsymbol{y}_{G,k})})
		)
		\\[0,2cm]
	&-
	m
	\log(
		F_{{\rm ELL}_1}(\tau;\, 0, 1+\boldsymbol{\lambda}^\top\boldsymbol{\Sigma}\boldsymbol{\lambda},g^{(1)})
		)
+
\sum_{k=1}^{m}
\sum_{i=1}^{n}
\log(G_i'(y_{ik})),
\end{align*}
where 	
$\boldsymbol{y}_{G,k}
=
(G_1(y_{1k}),\ldots,G_n(y_{nk}))^\top$. As $\boldsymbol{X}\sim {\rm ELL}_{n}(\boldsymbol{\mu},\boldsymbol{\Sigma},g^{(n)})$, by using formulas \eqref{eq:pdf:sym}, \eqref{id-main00} and \eqref{id-main10} in the above equation, the log-likelihood function (without the additive constant) is written as
\begin{align*}
\ell(\boldsymbol{\theta})
&=
{m\over 2}
\log(|\boldsymbol{\Sigma}^{-1}|)
+
\sum_{k=1}^{m}
\log(
g^{(n)} ((\boldsymbol{y}_{G,k} - \boldsymbol{\mu})^\top \boldsymbol{\Sigma}^{-1} (\boldsymbol{y}_{G,k} - \boldsymbol{\mu}))
)
\\[0,2cm]
&
+
\sum_{k=1}^{m}
\log\left(
\int_{-\infty}^{\boldsymbol{\lambda}^\top (\boldsymbol{y}_{G,k}-\boldsymbol{\mu})+\tau}
{
g^{(2)}(s^2+(\boldsymbol{y}_{G,k}-\boldsymbol{\mu})^\top \boldsymbol{\Sigma}^{-1}(\boldsymbol{y}_{G,k}-\boldsymbol{\mu}))
}
{\rm d}s
\right)
\\[0,2cm]
&-
\sum_{k=1}^{m}
\log
(
g^{(1)}
( 
(\boldsymbol{y}_{G,k}-\boldsymbol{\mu})^\top \boldsymbol{\Sigma}^{-1}(\boldsymbol{y}_{G,k}-\boldsymbol{\mu})
)
)
\\[0,2cm]
&+
{m\over 2}
\log(
1+\boldsymbol{\lambda}^\top\boldsymbol{\Sigma}\boldsymbol{\lambda})
-
m
\log\left(
\int_{-\infty}^{\tau} 
g^{(1)}\left({s^2\over 1+\boldsymbol{\lambda}^\top\boldsymbol{\Sigma}\boldsymbol{\lambda}}\right){\rm d}s
\right).
\end{align*}

The likelihood equations are given by
\begin{align*}
{\partial \ell(\boldsymbol{\theta})\over\partial\boldsymbol{\mu}}
=
\boldsymbol{0}_{n\times 1},
\quad
{\partial \ell(\boldsymbol{\theta})\over\partial\boldsymbol{\Sigma}^{-1}}
=
\boldsymbol{0}_{n\times n},
\quad 
{\partial \ell(\boldsymbol{\theta})\over\partial\boldsymbol{\lambda}}
=
\boldsymbol{0}_{n\times 1},
\quad 
	{\partial \ell(\boldsymbol{\theta})\over\partial\tau}
=0.
\end{align*}
In what follows we determine 
${\partial \ell(\boldsymbol{\theta})/\partial\boldsymbol{\mu}}$,
$
{\partial \ell(\boldsymbol{\theta})/\partial\boldsymbol{\Sigma}^{-1}}
$,
$
{\partial \ell(\boldsymbol{\theta})/\partial\boldsymbol{\lambda}}
$
and
$
{\partial \ell(\boldsymbol{\theta})/\partial\tau}
$. 
Indeed, by using the identities
\begin{align*}
{\partial \boldsymbol{a}^\top\boldsymbol{x}\over\partial \boldsymbol{x}}
=
\boldsymbol{a}^\top,
\quad
{\partial \boldsymbol{x}^\top\boldsymbol{A}\boldsymbol{x}\over\partial \boldsymbol{x}}
=2\boldsymbol{A}\boldsymbol{x},
\quad
{\partial \boldsymbol{x}^\top\boldsymbol{A}\boldsymbol{x}\over\partial \boldsymbol{A}}
=
\boldsymbol{x}\boldsymbol{x}^\top,
\quad
{\partial \boldsymbol{x}^\top\boldsymbol{A}^{-1}\boldsymbol{x}\over\partial \boldsymbol{A}}
=
-\boldsymbol{A}^{-\top}
\boldsymbol{x}\boldsymbol{x}^\top
\boldsymbol{A}^{-\top},
\quad
{\partial \log(\vert\boldsymbol{A}\vert)\over\partial \boldsymbol{A}}
=
\boldsymbol{A}^{-\top},
\end{align*}
with $\boldsymbol{A}$ being a $n\times n$ invertible matrix and $\boldsymbol{x}$ an $n$-dimensional vector,
we have
\begin{itemize}
\item [(i)]
\begin{align*}
{\partial \ell(\boldsymbol{\theta})\over\partial\boldsymbol{\mu}}
&=
-2\boldsymbol{\Sigma}^{-1} 
\sum_{k=1}^{m}
(\boldsymbol{y}_{G,k} - \boldsymbol{\mu})\,
{[g^{(n)}]' ((\boldsymbol{y}_{G,k} - \boldsymbol{\mu})^\top \boldsymbol{\Sigma}^{-1} (\boldsymbol{y}_{G,k} - \boldsymbol{\mu}))\over 
g^{(n)} ((\boldsymbol{y}_{G,k} - \boldsymbol{\mu})^\top \boldsymbol{\Sigma}^{-1} (\boldsymbol{y}_{G,k} - \boldsymbol{\mu}))
}
\\[0,2cm]
&-\boldsymbol{\lambda}^\top
\sum_{k=1}^{m}
{
	g^{(2)}([\boldsymbol{\lambda}^\top (\boldsymbol{y}_{G,k}-\boldsymbol{\mu})+\tau]^2+(\boldsymbol{y}_{G,k}-\boldsymbol{\mu})^\top \boldsymbol{\Sigma}^{-1}(\boldsymbol{y}_{G,k}-\boldsymbol{\mu}))
	\over 
\int_{-\infty}^{\boldsymbol{\lambda}^\top (\boldsymbol{y}_{G,k}-\boldsymbol{\mu})+\tau}
{
	g^{(2)}(s^2+(\boldsymbol{y}_{G,k}-\boldsymbol{\mu})^\top \boldsymbol{\Sigma}^{-1}(\boldsymbol{y}_{G,k}-\boldsymbol{\mu}))
}
{\rm d}s
}
\\[0,2cm]
&-
2\boldsymbol{\Sigma}^{-1}
\sum_{k=1}^{m}
(\boldsymbol{y}_{G,k}-\boldsymbol{\mu})\, 
{
	\int_{-\infty}^{\boldsymbol{\lambda}^\top (\boldsymbol{y}_{G,k}-\boldsymbol{\mu})+\tau}
{
	[g^{(2)}]'(s^2+(\boldsymbol{y}_{G,k}-\boldsymbol{\mu})^\top \boldsymbol{\Sigma}^{-1}(\boldsymbol{y}_{G,k}-\boldsymbol{\mu}))
}
{\rm d}s
	\over 
	\int_{-\infty}^{\boldsymbol{\lambda}^\top (\boldsymbol{y}_{G,k}-\boldsymbol{\mu})+\tau}
	{
		g^{(2)}(s^2+(\boldsymbol{y}_{G,k}-\boldsymbol{\mu})^\top \boldsymbol{\Sigma}^{-1}(\boldsymbol{y}_{G,k}-\boldsymbol{\mu}))
	}
	{\rm d}s
}
\\[0,2cm]
&+
2\boldsymbol{\Sigma}^{-1}
\sum_{k=1}^{m}
(\boldsymbol{y}_{G,k}-\boldsymbol{\mu})\, 
{
	[g^{(1)}]'
	( 
	(\boldsymbol{y}_{G,k}-\boldsymbol{\mu})^\top \boldsymbol{\Sigma}^{-1}(\boldsymbol{y}_{G,k}-\boldsymbol{\mu})
	)
	\over
g^{(1)}
( 
(\boldsymbol{y}_{G,k}-\boldsymbol{\mu})^\top \boldsymbol{\Sigma}^{-1}(\boldsymbol{y}_{G,k}-\boldsymbol{\mu})
)
},
\end{align*}

\item[(ii)] 
\begin{align*}
{\partial \ell(\boldsymbol{\theta})\over\partial\boldsymbol{\Sigma}^{-1}}
&=
{m\over 2}\, \boldsymbol{\Sigma}
+
\sum_{k=1}^{m}
(\boldsymbol{y}_{G,k} - \boldsymbol{\mu})
(\boldsymbol{y}_{G,k} - \boldsymbol{\mu})^\top \, 
\dfrac{[g^{(n)}]' ((\boldsymbol{y}_{G,k} - \boldsymbol{\mu})^\top \boldsymbol{\Sigma}^{-1} (\boldsymbol{y}_{G,k} - \boldsymbol{\mu}))}{
g^{(n)} ((\boldsymbol{y}_{G,k} - \boldsymbol{\mu})^\top \boldsymbol{\Sigma}^{-1} (\boldsymbol{y}_{G,k} - \boldsymbol{\mu}))
}
\\[0,2cm]
&
+
\sum_{k=1}^{m}
(\boldsymbol{y}_{G,k} - \boldsymbol{\mu})
(\boldsymbol{y}_{G,k} - \boldsymbol{\mu})^\top \, 
\dfrac{
	\int_{-\infty}^{\boldsymbol{\lambda}^\top (\boldsymbol{y}_{G,k}-\boldsymbol{\mu})+\tau}
	[g^{(2)}]'(s^2+(\boldsymbol{y}_{G,k}-\boldsymbol{\mu})^\top \boldsymbol{\Sigma}^{-1}(\boldsymbol{y}_{G,k}-\boldsymbol{\mu}))
	{\rm d}s
}{
\int_{-\infty}^{\boldsymbol{\lambda}^\top (\boldsymbol{y}_{G,k}-\boldsymbol{\mu})+\tau}
{
	g^{(2)}(s^2+(\boldsymbol{y}_{G,k}-\boldsymbol{\mu})^\top \boldsymbol{\Sigma}^{-1}(\boldsymbol{y}_{G,k}-\boldsymbol{\mu}))
}
{\rm d}s
}
\\[0,2cm]
&-
\sum_{k=1}^{m}
(\boldsymbol{y}_{G,k} - \boldsymbol{\mu})
(\boldsymbol{y}_{G,k} - \boldsymbol{\mu})^\top \, 
		{
	[g^{(1)}]'
	( 
	(\boldsymbol{y}_{G,k}-\boldsymbol{\mu})^\top \boldsymbol{\Sigma}^{-1}(\boldsymbol{y}_{G,k}-\boldsymbol{\mu})
	)
	\over 
	g^{(1)}
	( 
	(\boldsymbol{y}_{G,k}-\boldsymbol{\mu})^\top \boldsymbol{\Sigma}^{-1}(\boldsymbol{y}_{G,k}-\boldsymbol{\mu})
	)
}
\\[0,2cm]
&-
{m\over 2}\,
{
	\boldsymbol{\Sigma}\boldsymbol{\lambda} \boldsymbol{\lambda}^\top \boldsymbol{\Sigma}
	\over 
1+\boldsymbol{\lambda}^\top\boldsymbol{\Sigma}\boldsymbol{\lambda}
}
-
m\,
{	\boldsymbol{\Sigma}\boldsymbol{\lambda} \boldsymbol{\lambda}^\top \boldsymbol{\Sigma}\over (1+\boldsymbol{\lambda}^\top\boldsymbol{\Sigma}\boldsymbol{\lambda})^2}
\,
{
\int_{-\infty}^{\tau} 
s^2\,
[g^{(1)}]'\big({s^2\over 1+\boldsymbol{\lambda}^\top\boldsymbol{\Sigma}\boldsymbol{\lambda}} \big){\rm d}s
	\over 
\int_{-\infty}^{\tau} 
g^{(1)}\big({s^2\over 1+\boldsymbol{\lambda}^\top\boldsymbol{\Sigma}\boldsymbol{\lambda}}\big){\rm d}s
},
\end{align*}

\item[(iii)] 
\begin{align*}
{\partial \ell(\boldsymbol{\theta})\over\partial\boldsymbol{\lambda}}
&=
\sum_{k=1}^{m}
(\boldsymbol{y}_{G,k}-\boldsymbol{\mu})\,
\dfrac{
	g^{(2)}([\boldsymbol{\lambda}^\top (\boldsymbol{y}_{G,k}-\boldsymbol{\mu})+\tau]^2+(\boldsymbol{y}_{G,k}-\boldsymbol{\mu})^\top \boldsymbol{\Sigma}^{-1}(\boldsymbol{y}_{G,k}-\boldsymbol{\mu}))}{
\int_{-\infty}^{\boldsymbol{\lambda}^\top (\boldsymbol{y}_{G,k}-\boldsymbol{\mu})+\tau}
{
	g^{(2)}(s^2+(\boldsymbol{y}_{G,k}-\boldsymbol{\mu})^\top \boldsymbol{\Sigma}^{-1}(\boldsymbol{y}_{G,k}-\boldsymbol{\mu}))
}
{\rm d}s
}
\\[0,2cm]
&+
m\,
\dfrac{\boldsymbol{\Sigma}\boldsymbol{\lambda}}{
1+\boldsymbol{\lambda}^\top\boldsymbol{\Sigma}\boldsymbol{\lambda}}
+
2m\, 
{\boldsymbol{\Sigma}\boldsymbol{\lambda}\over (1+\boldsymbol{\lambda}^\top\boldsymbol{\Sigma}\boldsymbol{\lambda})^2}\,
\dfrac{\int_{-\infty}^{\tau} 
	s^2
	[g^{(1)}]'\big({s^2\over 1+\boldsymbol{\lambda}^\top\boldsymbol{\Sigma}\boldsymbol{\lambda}}\big)
	{\rm d}s}{
\int_{-\infty}^{\tau} 
g^{(1)}\big({s^2\over 1+\boldsymbol{\lambda}^\top\boldsymbol{\Sigma}\boldsymbol{\lambda}}\big){\rm d}s
},
\end{align*}

\item[(iv)]
\begin{align*}
	{\partial \ell(\boldsymbol{\theta})\over\partial\tau}
	=
	\sum_{k=1}^{m}
	\dfrac{		g^{(2)}([\boldsymbol{\lambda}^\top (\boldsymbol{y}_{G,k}-\boldsymbol{\mu})+\tau]^2+(\boldsymbol{y}_{G,k}-\boldsymbol{\mu})^\top \boldsymbol{\Sigma}^{-1}(\boldsymbol{y}_{G,k}-\boldsymbol{\mu}))
	}{
		\int_{-\infty}^{\boldsymbol{\lambda}^\top (\boldsymbol{y}_{G,k}-\boldsymbol{\mu})+\tau}
			g^{(2)}(s^2+(\boldsymbol{y}_{G.k}-\boldsymbol{\mu})^\top \boldsymbol{\Sigma}^{-1}(\boldsymbol{y}_{G,k}-\boldsymbol{\mu}))
		{\rm d}s
	}
-
m\,
\dfrac{g^{(1)}\big({\tau^2\over 1+\boldsymbol{\lambda}^\top\boldsymbol{\Sigma}\boldsymbol{\lambda}}\big)}{\int_{-\infty}^{\tau} 
	g^{(1)}\big({s^2\over 1+\boldsymbol{\lambda}^\top\boldsymbol{\Sigma}\boldsymbol{\lambda}}\big){\rm d}s}.
\end{align*}
\end{itemize}

No closed-form solution to the maximization problem is available. As such, the maximum likelihood (ML) estimator of $\bm\theta$, denoted by $\widehat{\boldsymbol{\theta}}$, can only be obtained via numerical optimization.  If $I(\boldsymbol{\theta}_0)$ denotes the expected Fisher information matrix, where $\boldsymbol{\theta}_0$ is the true value of the population parameter vector, then, under well-known regularity conditions \citep{Davison08}, it follows that
\begin{align}\label{asymptotic normality}
\sqrt{m}[I(\boldsymbol{\theta}_0)]^{1/2}(\widehat{\boldsymbol{\theta}}-\boldsymbol{\theta}_0)
\stackrel{d}{\longrightarrow} N(\boldsymbol{0}_{(n+1)^2\times 1},I_{(n+1)^2\times(n+1)^2}),
\quad 
\text{as} \ m\to\infty,
\end{align}
where $\boldsymbol{0}_{(n+1)^2\times 1}$ is the $(n+1)^2\times$ zero vector, and $I_{(n+1)^2\times(n+1)^2}$
is the ${(n+1)^2\times(n+1)^2}$ identity matrix. Since the  expected Fisher information can be approximated by its observed version (obtained from the Hessian matrix), we can use the diagonal elements of this observed version to approximate the standard
errors of the ML estimates.

Note that, for $\bm \lambda=0$ and $\tau=0$, the multivariate extended $G$-skew-normal belongs to the exponential family.
This is easy to verify because, in this case, the \text{EGSE}$_n$ PDF in \eqref{def-pdf-T*},
with $g^{(n)}(x) = \exp(-x/2)$ and $Z_{g^{(n)}}= 2\pi$, can be expressed as
\begin{align*}
	f_{\boldsymbol{Y}}(\boldsymbol{y})
	&=
		\frac{1}{2\pi |\boldsymbol{\Sigma}|^{1/2}}\, 
	\exp\left(
	-{1\over 2}\,
\boldsymbol{y}_{G}^\top\boldsymbol{\Sigma}^{-1}\boldsymbol{y}_{G}
+
\boldsymbol{y}_{G}^\top\boldsymbol{\Sigma}^{-1}\boldsymbol{\mu}
	-{1\over 2}\,
\boldsymbol{\mu}^\top \boldsymbol{\Sigma}^{-1}\boldsymbol{\mu}
	\right)
	\, \prod_{i=1}^n G_i'(y_i)
		\\[0,2cm]
	&=
	H(\boldsymbol{y})
	\exp\left(
	S^\top(\bm\theta) T(\boldsymbol{y})-\psi(\boldsymbol{\theta})\right),
	\quad 
	\boldsymbol{y}\in D^n,
\end{align*}
where $\boldsymbol{\Sigma}^{-1}\equiv (\sigma_{ij}^{-1})_{n\times n}$ is the inverse matrix of $\boldsymbol{\Sigma}$, $H(\boldsymbol{y})=\prod_{i=1}^n G_i'(y_i)$, $\psi(\boldsymbol{\theta})=\boldsymbol{\mu}^\top \boldsymbol{\Sigma}^{-1}\boldsymbol{\mu}/2+\log(2\pi |\boldsymbol{\Sigma}|^{1/2})$,
\begin{align*}
T(\boldsymbol{y})
=
(\{G_i(y_i)\}_{i=1,\ldots,n},\ldots,\{G_i^2(y_i)\}_{i=1,\ldots,n}, \{G_i(y_i)G_j(y_j)\}_{1\leqslant i<j\leqslant n})^\top
\end{align*}
and
\begin{align*}
S(\bm\theta)
=
\left(\left\{\sum_{j=1}^{n}\mu_j\sigma_{ij}^{-1}\right\}_{i=1,\ldots,n},\left\{-{1\over 2}\,\sigma_{ii}^{-1}\right\}_{i=1,\ldots,n}, \{-\sigma_{ij}^{-1}\}_{1\leqslant i<j\leqslant n}\right)^\top.
\end{align*}
For distributions belongs to the exponential family
the asymptotic normality in \eqref{asymptotic normality} follows by applying Theorem 6.1 of \cite{Berk72}.

\section{Simulation study}\label{sec:simulations}

In this section, a simulation study is conducted for evaluating the performance of the maximum likelihood estimators. The simulation study considers the estimation of model parameters in the bivariate case.  For illustrative purposes, we only present the  results for the extended unit-$G$-skew-normal distribution (due to
space limitations we omit the results of the extended unit-$G$-skew-Student-$t$ distribution) with two $G_i$ functions: \( G_i(x) = \tan\left((x - {1}/{2})\pi\right)\) and \(G_i(x) = \log\left({x^3}/{(1 - x^3)} \right)\); see Table \ref{table:2-1}.

The performance and recovery of the maximum likelihood estimators are evaluated by means of the relative bias (RB) and the root mean square error (RMSE), given by
\begin{eqnarray*}
 \widehat{\textrm{RB}}(\widehat{\theta}) &=&  \frac{1}{N} \sum_{i = 1}^{N}\left| \frac{(\widehat{\theta}^{(i)} - \theta)}{\theta}\right| ,\quad
\widehat{\mathrm{RMSE}}(\widehat{\theta}) = {\sqrt{\frac{1}{N} \sum_{i = 1}^{N} (\widehat{\theta}^{(i)} - \theta)^2}},
\end{eqnarray*}
where $\theta$ and $\widehat{\theta}^{(i)}$ are the true parameter value and its $i$-th estimate, and $N$ is the number of Monte Carlo replications. The simulation scenario considered is as follows: the sample size varies between \( n \in \{200, 500, 1000, 2000\} \), with the true parameters defined as $$(\mu_1, \mu_2, \lambda_1, \lambda_2, \tau, \sigma_1, \sigma_2)^\top = (1,1,0.5,0.6,0.5,1,1)^\top,$$ and \(\rho\) assuming values \(\{0.10, 0.25, 0.50, 0.75, 0.90 \}\). In all cases, 100 Monte Carlo replications were performed for each setting.


Figures \ref{Figure 3.1}--\ref{Figure 3.4} show maximum likelihood estimation results.  From these figures, it is possible to observe a clear convergence of the RB towards zero for all parameters as sample sizes increase. This pattern is also evident when analyzing the RMSE, indicating a decrease in the corresponding variance as the sample size increases. From Figure \ref{Figure 3.2}, it is observed that the RMSE of $\widehat{\lambda}_{1}$ does not consistently decrease across all possibilities for $\rho$. Several factors may influence this behavior, such as the sample size, the number of iterations, or the inverse transformation $G^{-1}_i$ used.

\begin{figure}[H]
  \centering
  \label{fig_mc_01}
    \includegraphics[width=0.85\linewidth]{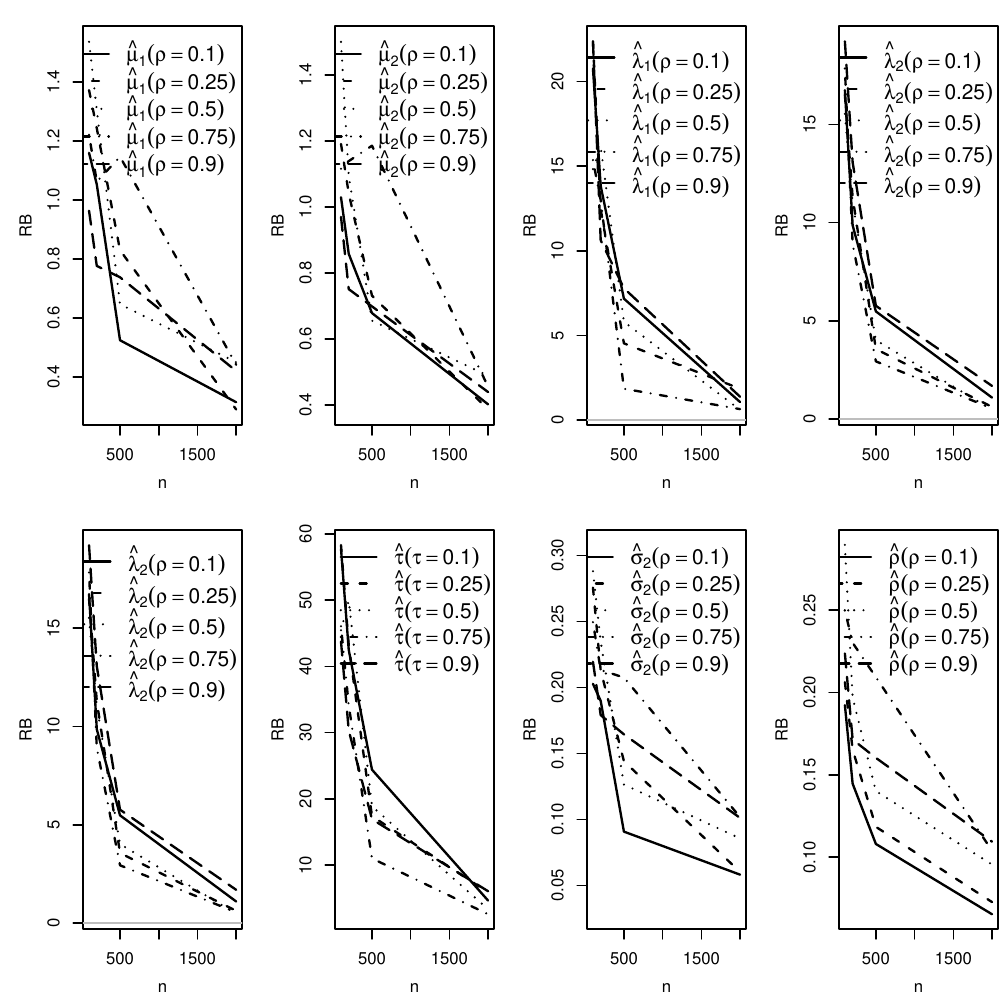}
    \caption{Relative bias for $G_{i}^{-1}(x) = \frac{1}{2} + \frac{ \arctan(x)}{\pi}$.}
    \label{Figure 3.1}
\end{figure}

\newpage
\begin{figure}[H]
    \centering
    \label{fig_mc_02}
    \includegraphics[width=0.85\linewidth]{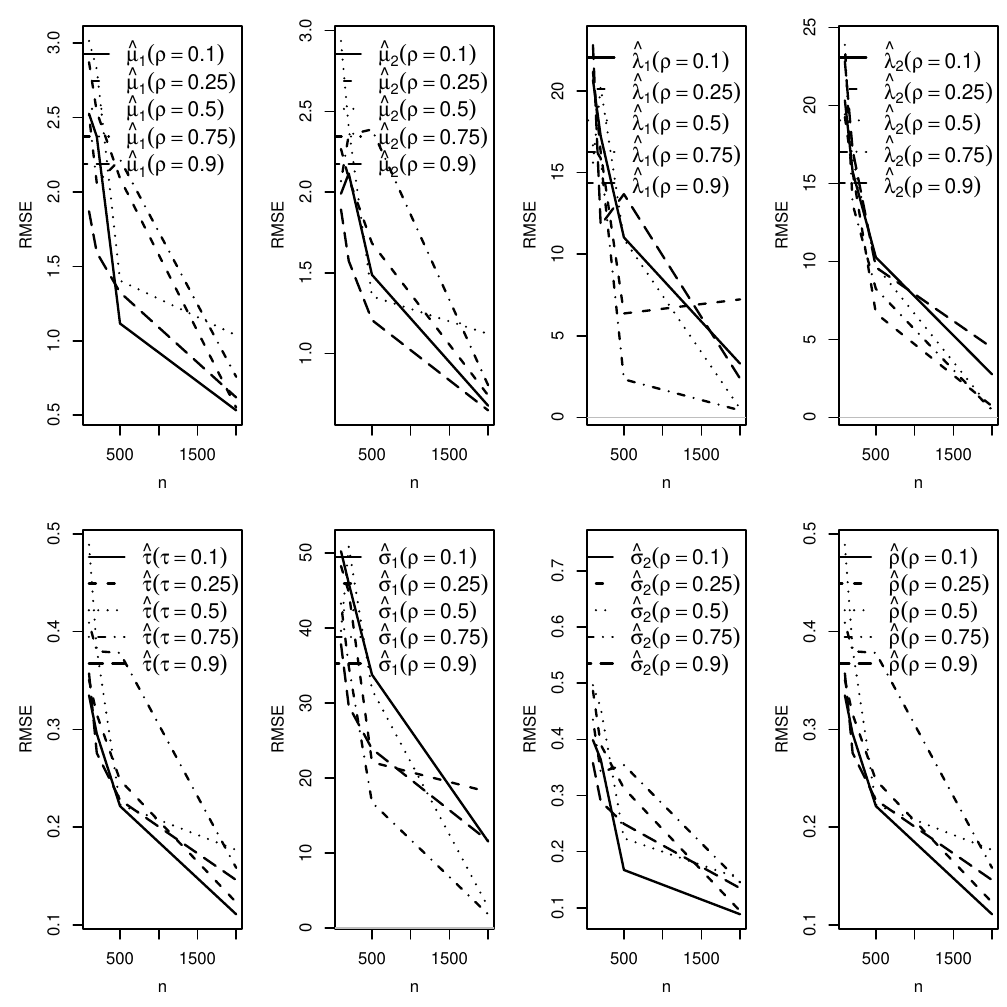}
    \caption{Root mean squared error for $G_{i}^{-1}(x) = \frac{1}{2} + \frac{ \arctan(x)}{\pi}.$}
    \label{Figure 3.2}
\end{figure}


\newpage
\begin{figure}[H]
  \centering
  \label{fig_mc_03}
    \includegraphics[width=0.85\linewidth]{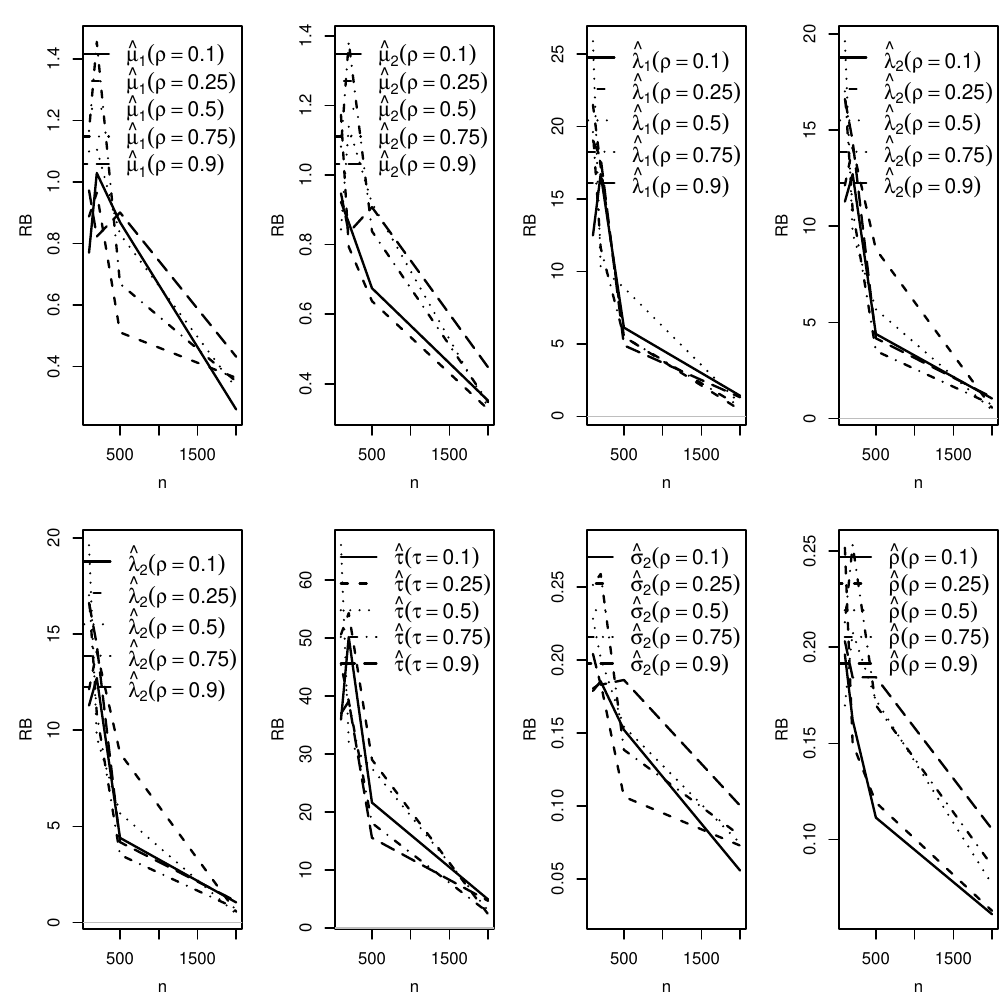}
    \caption{Relative bias for $G^{-1}_{i}(x) = \big[\frac{ \exp(x)}{1 + \exp(x)} \big]^{\frac{1}{3}}$.}
    \label{Figure 3.3}
\end{figure}

\newpage
\begin{figure}[H]
  \centering
  \label{fig_mc_04}
    \includegraphics[width=0.85\linewidth]{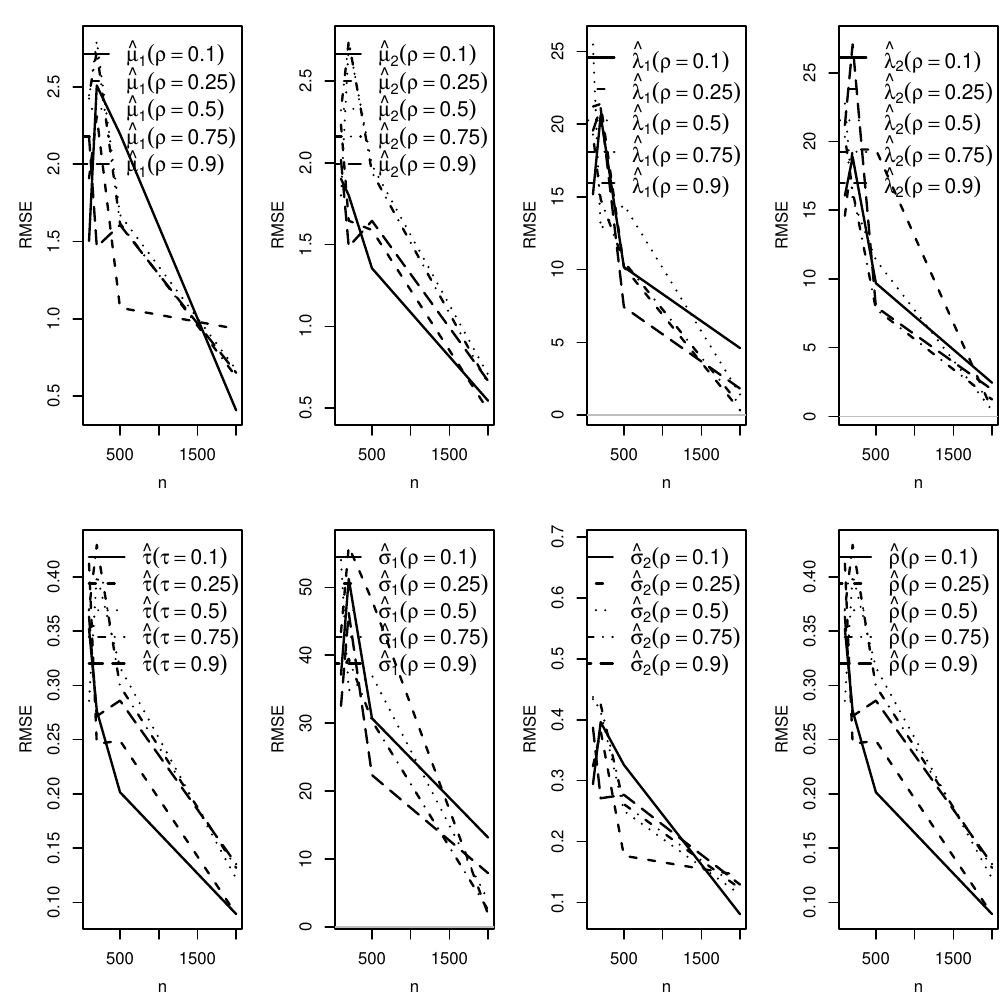}
    \caption{Root mean squared error for $G^{-1}_{i}(x) = \big[\frac{ \exp(x)}{1 + \exp(x)} \big]^{\frac{1}{3}}$.}
    \label{Figure 3.4}
\end{figure}


\newpage
\section{Application to real data}\label{sec:applications}

\quad \, In this section, we illustrate the proposed model and the inferential method using real data on socioeconomic indicators for each of Switzerland's 47 French-speaking provinces in 1888. This data set is called \textit{swiss} and is available in the R software. The aim of the study was to explore the relationships between fertility (measured as the birth rate) and several other socioeconomic variables in 47 districts. The variables contained in the dataset are:
\begin{itemize}
    \item \textit{Fertility}: Fertility rate (average number of births per 1000 women).
    \item \textit{Agriculture}: Percentage of men involved in agricultural activities.
    \item \textit{Examination}: Percentage of military draftees \textit{draftees} who received a high score on aptitude exams.
    \item \textit{Education}: Percentage of men with education beyond primary education.
    \item \textit{Catholic}: Percentage of Catholics (as a measure of religion and tradition).
    \item \textit{Infant.Mortality}: Infant mortality rate (number of baby deaths per 1000 live births).
\end{itemize}

For the application presented here, the variables Education and Agriculture were considered. The data can be found at \href{https://stat.ethz.ch/R-manual/R-devel/library/datasets/html/swiss.html}{Swiss Fertility and Socioeconomic Indicators (1888).}

Table \ref{tab:summary_stats} presents the descriptive statistics of the two variables: Education and Agriculture, both with a set of 47 observations. For the Education variable, it is observed that the minimum value recorded is 0.010, while the maximum reaches 0.530, with a median of 0.080 and an average of 0.1098. The dispersion of the Education data is reflected by the standard deviation (SD) of 0.0962, which suggests considerable variation in relation to the mean. This is further evidenced by the coefficient of variation (CV) of 87.5822, indicating a high relative variability of the data. Positive skewness, with a skewness coefficient (CS) of 2.3428, suggests that the data distribution is skewed to the right, which is reinforced by the kurtosis coefficient (CK) of 6.5414, indicating a more elongated distribution with heavy tails. Considering the Agriculture variable, the minimum value is 0.012 and the maximum is 0.897, with a median of 0.541, very close to the average of 0.5066, which suggests a more balanced distribution. The standard deviation is higher, 0.2271, reflecting greater data dispersion compared to Education. The coefficient of variation is 44.8311, less high than that of Education, suggesting less relative variability. The Agriculture distribution presents negative skewness, with an asymmetry coefficient of -0.3309, indicating a slight leftward bias. The negative kurtosis coefficient (-0.7926) suggests a flatter distribution with lighter tails, in contrast to the more elongated distribution of Education.

\begin{table}[h!]
\centering
\begin{tabular}{lcccccccccc}
\hline
\textbf{Variables} & \textbf{n} & \textbf{Minimum} & \textbf{Median} & \textbf{Mean} & \textbf{Maximum} & \textbf{SD} & \textbf{CV} & \textbf{CS} & \textbf{CK} \\
\hline
\textbf{Education} & 47 & 0.01 & 0.08 & 0.11 & 0.53 & 0.096 & 87.58 & 2.33 & 6.54 \\
\textbf{Agriculture} & 47 & 0.012 & 0.54 & 0.51 & 0.9 & 0.23 & 44.83 & -0.33 & -0.79 \\
\hline
\end{tabular}
\caption{Summary statistics.}
\label{tab:summary_stats}
\end{table}

The extended unit-$G$-skew-normal and extended unit-$G$-skew-Student-$t$ distributions were used to fit the data. We considered the $G_i$ functions with domain $D\in(0,1)$; see Table \ref{table:2-1}. The model parameters were estimated according to the methodology presented in Section \ref{sec:mle} -- for simplification purposes $\tau$ was set to zero. The estimation of the $\nu$ parameter of the extended unit-$G$-skew-Student-$t$ distribution was carried out by using the profile likelihood method. First, an initial grid of values was defined for $\nu\in\{1,2,\ldots,50\}$, then for each fixed value of $\nu$ it is computed the maximum likelihood estimates of the remaining parameters and also the log-likelihood function. The final estimate of $\nu$ is the one that maximizes the log-likelihood function and the associated estimates of the remaining parameters are then the final ones; see \cite{sauloetal:2021}.

Tables \ref{tab: FDIC 1}-\ref{tab: FDIC 4} report the Kolmogorov-Smirnov (KS) and Anderson-Darling (AD) tests, the maximum likelihood estimates, and the standard errors for the extended unit-$G$-skew-normal and extended unit-$G$-skew-Student-$t$ distributions. Moreover, Figures \ref{fig:matriz_imagens1}-\ref{fig:matriz_imagens3} display the quantile versus quantile (QQ) plots of the randomized quantile \citep{sauloetal:2022} residuals for these models. From these results, we observe that the extended unit-$G$-skew-normal
model provides better adjustment compared to the unit-$G$-skew-Student-$t$ model. Note that the results of the QQ plots indicate that $G_{i}(x) = \log({x}/({1-x}))$ shows better agreement with the expected standard normal distribution; note also that the p-values of the KS and AD tests favor the extended unit-$G$-skew-normal with $G_{i}(x) = \log({x}/({1-x}))$.

\begin{table}[!ht]
\small
\centering
\caption{KS and AD test results.}
\begin{tabular}{lccccc}
\hline
\multicolumn{3}{c}{\textbf{Extended unit-$G$-skew-Student-$t$}} \\
\hline
$G_i(x)$  & p-value.KS & p-value.AD \\
  \hline
$\tan(\pi(x -\frac{1}{2}))$  & 0.18 & 0.08  \\
  \hline
$\log(\frac{x^{3}}{1-x^{3}})$ & 0.18 & 0.07  \\
  \hline
 $ \log(\frac{x^5}{1-x^5})$  & 0.18 & 0.02 \\
   \hline
$\log(-\log(1-x))$  & 0.17 & 0.03 \\
  \hline
$- \log(1-x)$  & 0.05 & 0.02 \\
  \hline
$1- \log(-\log(x))$ & 0.18 & 0.04 \\
  \hline
$ \log(\log(\frac{1}{-x+1})+1)$  & 0.00 & 0.00 \\
  \hline
 $ \log( \frac{x}{1-x})$  & 0.16 & 0.03 \\
   \hline
\end{tabular}
\label{tab: FDIC 1}
\end{table}

\begin{table}[!ht]
\small
\centering
\caption{KS and AD test results.}
\begin{tabular}{lccccc}
\hline
\multicolumn{3}{c}{\textbf{Extended unit-$G$-skew-normal}} \\
\hline
$G_i(x)$  & p-value.KS & p-value.AD \\
  \hline
$\tan(\pi(x -\frac{1}{2}))$ & 0.03 & 0.01  \\
  \hline
$\log(\frac{x^{3}}{1-x^{3}})$  & 0.23 & 0.03 \\
  \hline
 $ \log(\frac{x^5}{1-x^5})$  & 0.23 & 0.04 \\
   \hline
$\log(-\log(1-x))$  & 0.35 & 0.03 \\
  \hline
$- \log(1-x)$  & 0.24 & 0.08  \\
  \hline
$1- \log(-\log(x))$  & 0.35 & 0.06 \\
  \hline
$ \log(\log(\frac{1}{-x+1})+1)$  & 0.00 & 0.00 \\
  \hline
 $ \log( \frac{x}{1-x})$ & 0.35 & 0.05  \\
   \hline
\end{tabular}

\label{tab: FDIC 2}
\end{table}

\begin{table}[!ht]
\small
\centering
\caption{ Parameters estimates (with standard errors in parentheses).}
\begin{tabular}{lccccccccc}
\hline
\multicolumn{9}{c}{\textbf{Extended unit-$G$-skew-Student-$t$}} \\
\hline
 $G_{i}(x)$ & $\hat{\mu}_1$ & $\hat{\mu}_{2}$ & $\hat{\lambda}_1$ & $\hat{\lambda}_2$ & $\hat{\sigma}_1$ & $\hat{\sigma}_2$ & $\hat{\rho}$ &  $\hat{\nu}$\\
  \hline
$\tan(\pi(x -\frac{1}{2}))$  & -1.63 & -0.06 & -2.23 & -2.72 & 3.77 & 0.85 & -0.31 & 2\\
   & (0.41) & (0.27) & (0.94) & (1.57) & (0.87) & (0.14) & (0.30) & -\\
   \hline
   $\log(\frac{x^{3}}{1-x^{3}})$
 & -4.68 & -4.10 & -0.65 & -0.10 & 4.53 & 4.01 & -0.88 & 31 \\
    & (1.04) & (1.67) & (0.29) & (0.28) & (1.96) & (2.31) & (0.13) & - \\
  \hline
$ \log(\frac{x^5}{1-x^5})$  & -5.21 & -8.19 & -0.92 & -0.20 & 10.45 & 6.89 & -0.92 & 16 \\
  & (1.56) & (1.85) & (0.87) & (0.20) & (3.28) & (2.84) & (0.06) & -\\
  \hline
$\log(-\log(1-x))$ & -1.46 & -0.61 & -5.51 & -3.22 & 1.34 & 0.90 & -0.49 & 46 \\
 & (0.22) & (0.39) & (3.05) & (1.61) & (0.11) & (0.05) & (0.28) & - \\
  \hline
$- \log(1-x)$  & 0.12 & 0.62 & 1.51 & 1.47 & 0.08 & 0.50 & -0.55 & 8 \\
   & (0.02) & (0.26) & (7.39) & (1.52) & (0.01) & (0.08) & (0.14) & - \\
    \hline
$1- \log(-\log(x))$  & 0.08 & 1.52 & 0.39 & -0.08 & 0.32 & 0.71 & -0.67 & 15 \\
   & (0.25) & (0.51) & (3.46) & (1.91) & (0.03) & (0.08) & (0.10) & -\\
 \hline
$ \log(\log(\frac{1}{-x+1})+1)$ & 0.04 & 0.93 & 0.73 & 0.19 & -0.10 & 0.46 & 0.76 & 23\\
    & (0.02) & (0.06) & (2.25) & (0.32) & (0.01) & (0.01) & (0.03) & - \\
    \hline
$ \log( \frac{x}{1-x})$ & -3.12 & 1.20 & 0.26 & -1.06 & 1.18 & 1.72 & -0.84 & 24 \\
   & (0.40) & (0.34) & (1.15) & (0.91) & (0.34) & (0.37) & (0.10) & - \\
   \hline
\end{tabular}
\label{tab: FDIC 3}
\end{table}

\begin{table}[!ht]
\small
\centering
\caption{ Parameters estimates (with standard errors in parentheses).}
\begin{tabular}{lccccccccc}
\hline
\multicolumn{9}{c}{\textbf{Extended unit-$G$-skew-normal}} \\
\hline
 $G_{i}(x)$ & $\hat{\mu}_1$ & $\hat{\mu}_{2}$ & $\hat{\lambda}_1$ & $\hat{\lambda}_2$ & $\hat{\sigma}_1$ & $\hat{\sigma}_2$ & $\hat{\rho}$ \\
  \hline
$\tan(\pi(x -\frac{1}{2}))$  & -1.26 & 0.32 & -2.75 & -3.02 & 6.67 & 3.75 & -0.14 \\
  & (0.39) & (0.51) & (2.41) & (3.80) & (0.63) & (0.35) & (0.14) \\
   \hline
   $\log(\frac{x^{3}}{1-x^{3}})$ & -3.88 & -4.36 & -1.12 & -0.42 & 6.18 & 4.90 & -0.91   \\
    & (0.12) & (0.69) & (0.44) & (0.27) & (1.47) & (1.57) & (0.06)  \\
  \hline
$ \log(\frac{x^5}{1-x^5})$  & -5.40 & -6.97 & -2.02 & -0.43 & 9.66 & 4.76 & -0.78   \\
  & (0.52) & (0.96) & (1.75) & (0.33) & (1.51) & (0.78) & (0.10)  \\
  \hline
$\log(-\log(1-x))$ & -2.59 & 0.14 & -0.62 & -1.57 & 0.79 & 1.08 & -0.58   \\
 &( 0.70) & (1.27) & (0.90) & (3.60) & (0.07) & (0.65) & (0.06)   \\
  \hline
$- \log(1-x)$  & 0.14 & 0.67 & -0.05 & 0.71 & 0.13 & 0.55 & -0.55   \\
   & (0.05) & (0.20) & (3.88) & (1.09) & (0.02) & (0.01) & (0.13)    \\
    \hline
$1- \log(-\log(x))$  & 0.34 & 1.01 & -0.75 & 0.58 & 0.42 & 0.91 & -0.78    \\
   & (0.12) & (0.49) & (1.57) & (1.61) & (0.07) & (0.23) & (0.02)  \\
 \hline
$ \log(\log(\frac{1}{-x+1})+1)$ & 0.06 & 0.93 & -0.23 & 0.30 & -0.17 & 0.87 & 0.88   \\
    & (0.15) & (0.79) & (1.72) & (5.11) & (0.52) & (3.58) & (0.80)   \\
    \hline
$ \log( \frac{x}{1-x})$ & -2.36 & 0.02 & -0.14 & -0.12 & 0.89 & 1.21 & -0.71  \\
   & (1.05) & (1.02) & (3.02) & (1.89) & (0.09) & (0.12) & (0.02)   &\\
   \hline
\end{tabular}
\label{tab: FDIC 4}
\end{table}

\begin{figure}[!ht]
\centering
\subfigure[Extended unit-$G$-skew-Student-$t$ with $G_{i}(x) = \tan ((x-{1\over 2})\pi)$.]{\includegraphics[scale=0.5]{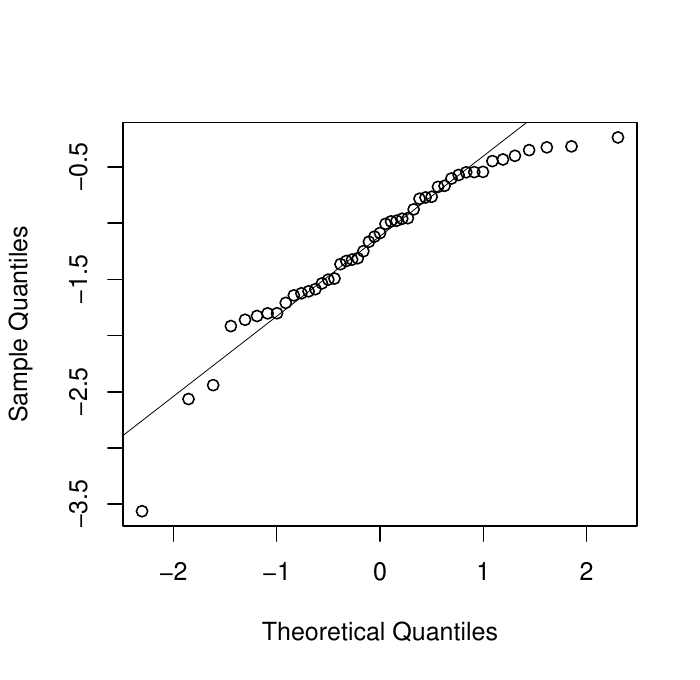}}\hspace{5mm}
\subfigure[Extended unit-$G$-skew-normal with $G_{i}(x) = \tan ((x-{1\over 2})\pi)$.]
{\includegraphics[scale=0.5]{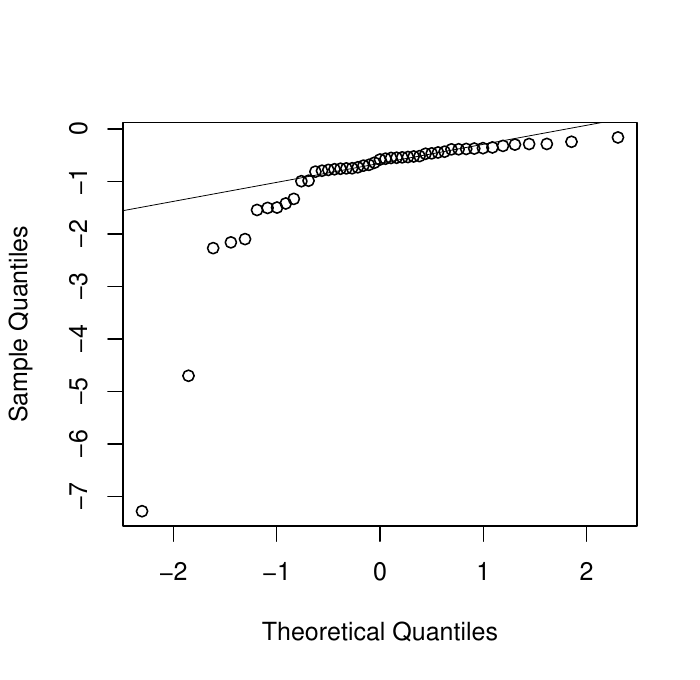}}\\
\subfigure[Extended unit-$G$-skew-Student-$t$ with $G_i(x) = \log (\frac{x^{3}}{1-x^{3}} )$.]{\includegraphics[scale=0.5]{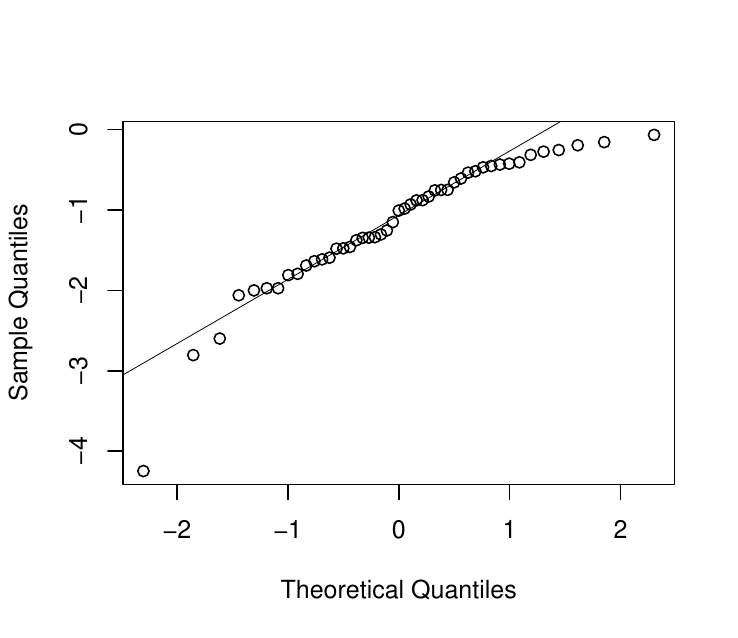}}\hspace{5mm}
\subfigure[Extended unit-$G$-skew-normal with $G_i(x) =\log (\frac{x^{3}}{1-x^{3}} )$.]{\includegraphics[scale=0.5]{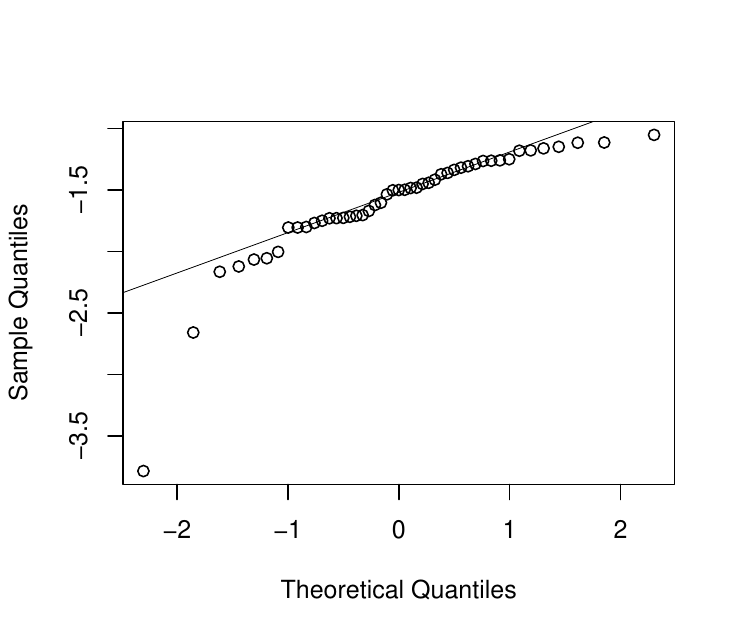}}
 \caption{{QQ plot of randomized quantile residuals for the indicated models.}}
\label{fig:matriz_imagens1}
\end{figure}

\begin{figure}[!ht]
\centering
\subfigure[Extended unit-$G$-skew-Student-$t$ with $G_{i}(x) =  \log (\frac{x^5}{1-x^5})$.]{\includegraphics[scale=0.5]{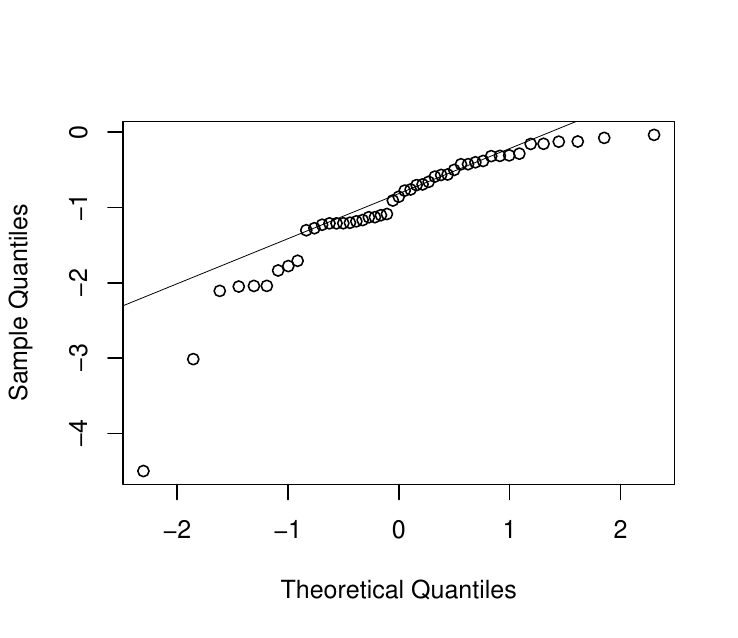}}\hspace{5mm}
\subfigure[Extended unit-$G$-skew-normal  with $G_{i}(x) = \log (\frac{x^5}{1-x^5})$.]{\includegraphics[scale=0.5]{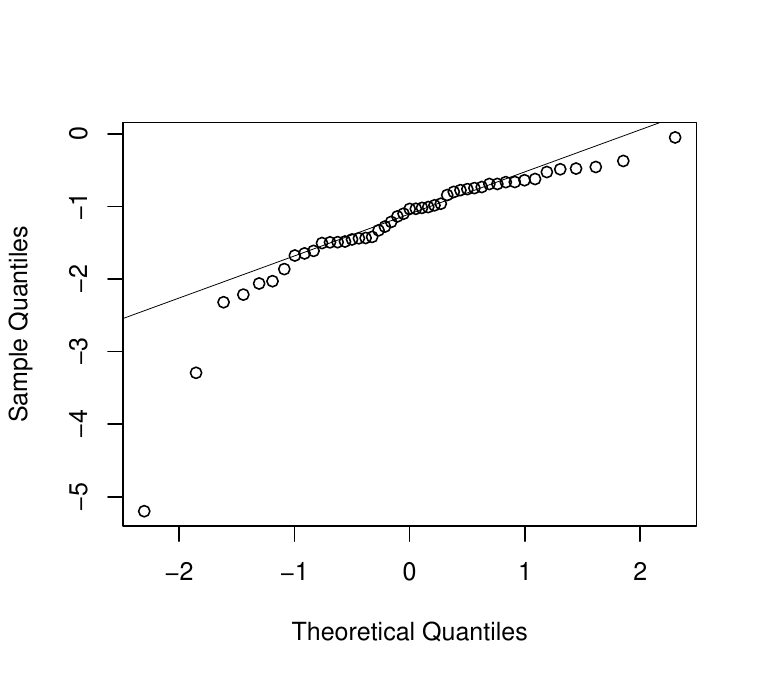}}\\
\subfigure[Extended unit-$G$-skew-Student-$t$  with $G_{i}(x) =\log(\log(\frac{1}{-x+1})+1)$.]
{\includegraphics[scale=0.5]{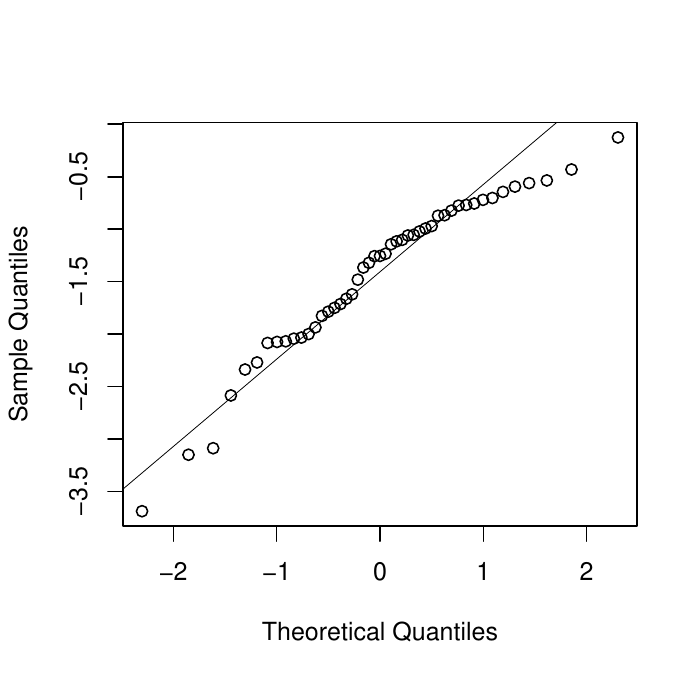}}\hspace{5mm}
\subfigure[Extended unit-$G$-skew-normal with $G_{i}(x) =\log(\log(\frac{1}{-x+1})+1)$.]
{\includegraphics[scale=0.5]{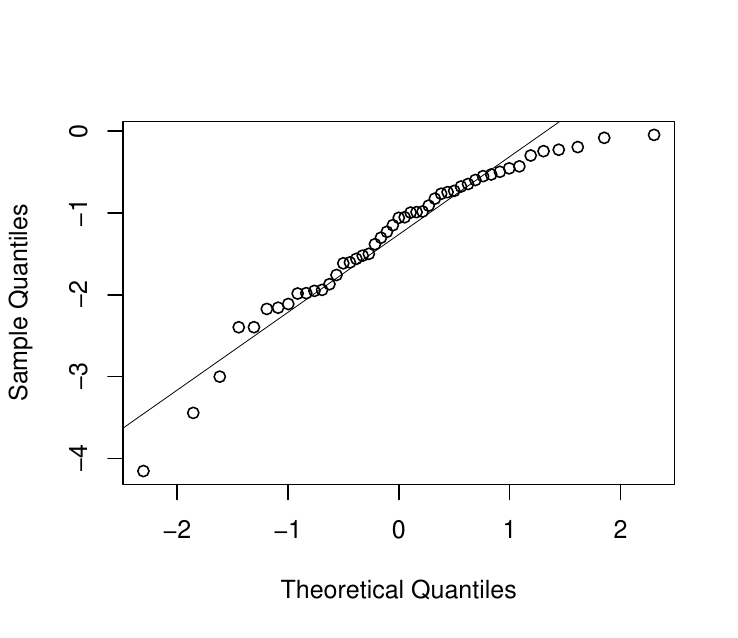}}\\
\subfigure[Extended unit-$G$-skew-Student-$t$ with $G_{i}(x) = - \log(1-x)$.]
{\includegraphics[scale=0.5]{Images/qq/f5_student.pdf}}\hspace{5mm}
\subfigure[Extended unit-$G$-skew-normal with $G_{i}(x) = - \log(1-x)$.]
{\includegraphics[scale=0.5]{Images/qq/f5_normal.pdf}}
 \caption{{QQ plot of randomized quantile residuals for the indicated models.}}
\label{fig:matriz_imagens2}
\end{figure}

\begin{figure}[!ht]
\centering
\subfigure[Extended unit-$G$-skew-Student-$t$  with $G_{i}(x) = 1-\log(-\log(x))$.]
{\includegraphics[scale=0.5]{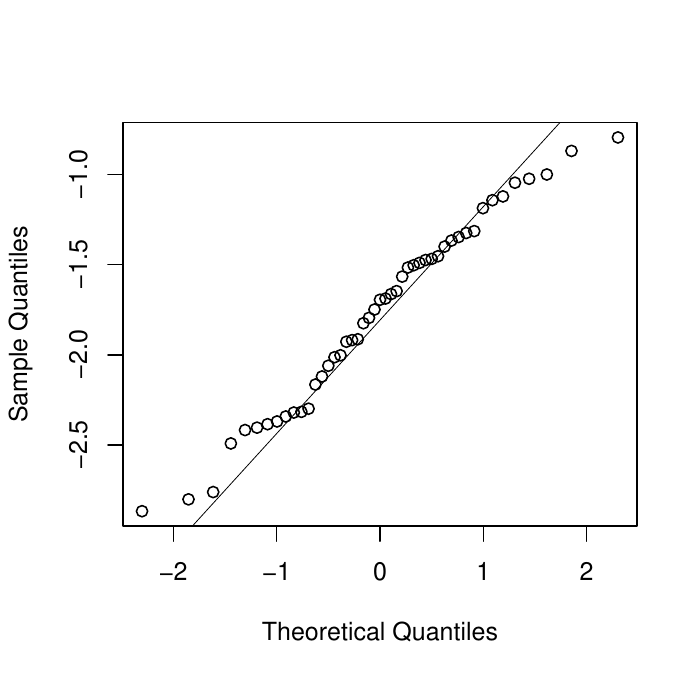}}\hspace{5mm}
\subfigure[Extended unit-$G$-skew-normal  with $G_{i}(x)= 1-\log(-\log(x))$.]
{\includegraphics[scale=0.5]{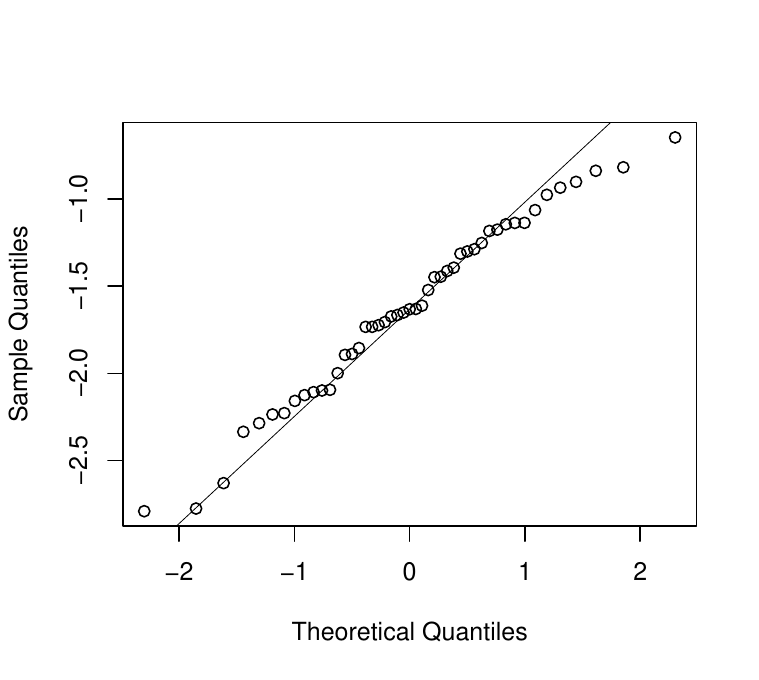}}\\
\subfigure[Extended unit-$G$-skew-Student-$t$ with $G_{i}(x) = \log(-\log(1-x))$.]
{\includegraphics[scale=0.5]{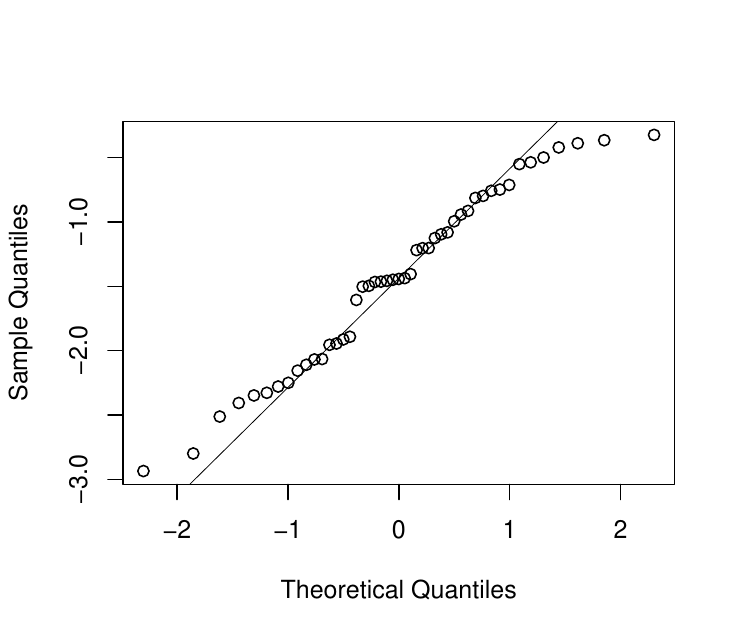}}\hspace{5mm}
\subfigure[Extended unit-$G$-skew-normal  with $G_{i}(x) =\log(-\log(1-x))$.]
{\includegraphics[scale=0.5]{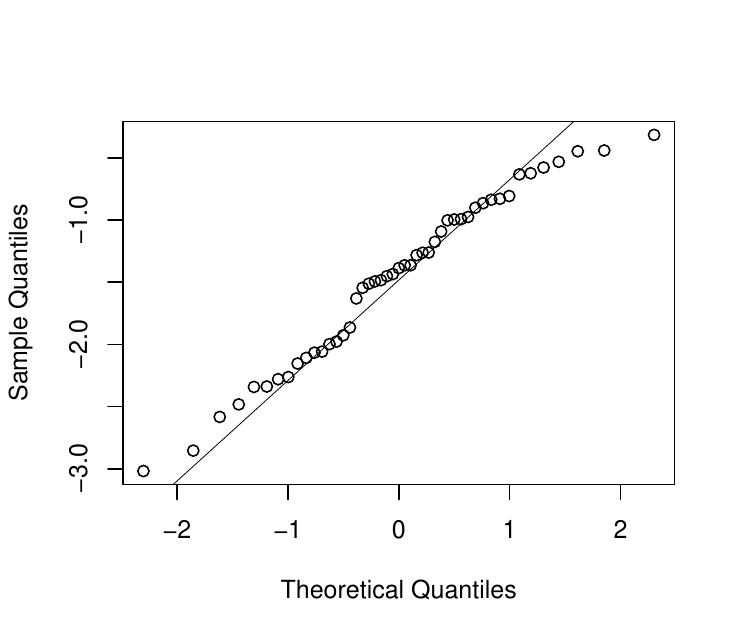}}\\
\subfigure[Extended unit-$G$-skew-Student-$t$ with $G_{i}(x) =  \log( \frac{x}{1-x})$.]
{\includegraphics[scale=0.5]{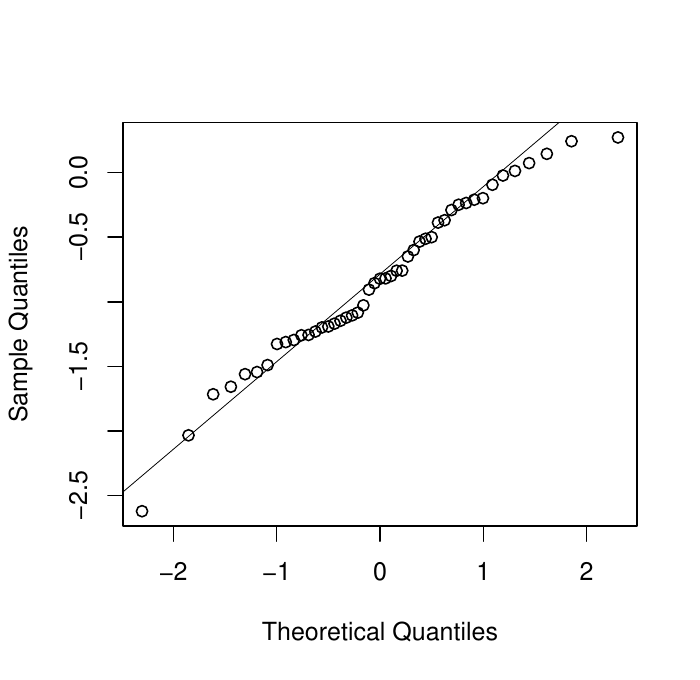}}\hspace{5mm}
\subfigure[Extended unit-$G$-skew-normal with $G_{i}(x) = \log( \frac{x}{1-x})$.]
{\includegraphics[scale=0.5]{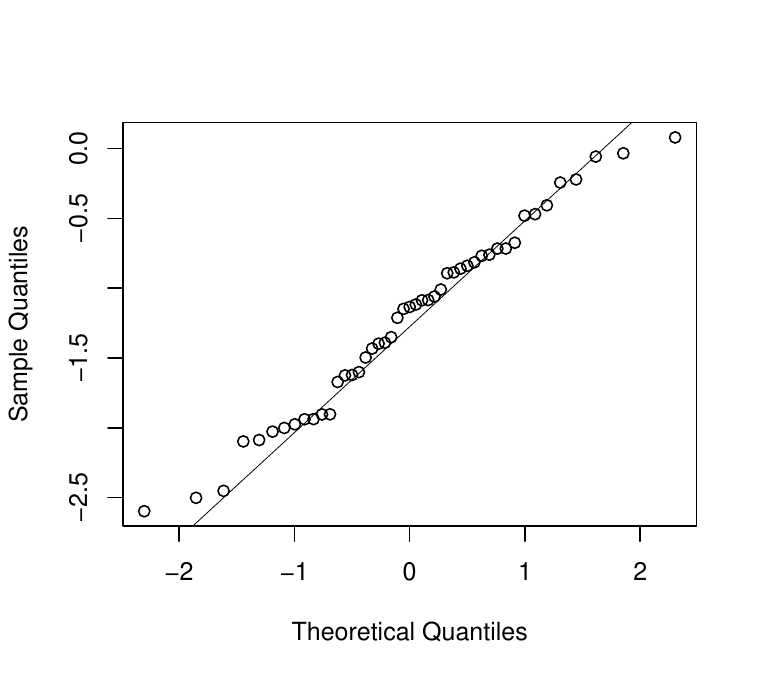}}
 \caption{{QQ plot of randomized quantile residuals for the indicated models.}}
\label{fig:matriz_imagens3}
\end{figure}

\clearpage

\section{Concluding Remarks}\label{sec:conclusions}
\noindent

In this paper, we introduced a family of multivariate asymmetric distributions over an arbitrary subset of set of real numbers, based on commonly used elliptically symmetric distributions. We have discussed several theoretical properties such as (non-)identifiability, quantiles, stochastic representation, conditional and marginal distributions, moments, and parameter estimation.  A Monte Carlo simulation study has been carried out for evaluating the performance of the maximum likelihood estimates. The simulation results show that the estimators perform very well, with relative bias and root mean square error being close to zero. We have applied the proposed models to a real socioeconomic data set,  and the results has favored the use of the extended unit-$G$-skew-normal model over the unit-$G$-skew-Student-$t$ model.

\paragraph{Acknowledgements}
The authors gratefully acknowledge financial support from CNPq, CAPES and FAP-DF, Brazil.

\paragraph{Disclosure statement}
There are no conflicts of interest to disclose.

\normalsize


%
%
%
%

\end{document}